\documentclass[preprint,numbers,9pt,nocopyrightspace]{sigplanconf}

\newcommand{\todo}[1]{[To do: #1]}
\newcommand{\hide}[1]{}

\usepackage{amsmath,amssymb,amsthm}
\usepackage{stmaryrd}
\usepackage{proof}
\usepackage{xspace}
\usepackage{xy}
\xyoption{all}

\theoremstyle{plain}

\newtheorem{theorem}{Theorem}[section]

\newtheorem{lemma}[theorem]{Lemma}
\newtheorem{proposition}[theorem]{Proposition}

\theoremstyle{definition}
\newtheorem{definition}[theorem]{Definition}
\newtheorem{assumption}[theorem]{Assumption}

\newcommand*{\R}{\ensuremath{\mathbb{R}}}

\newcommand*{\Rpz}{\ensuremath{\mathbb{R}_{{\geq}0}}}

\newcommand*{\N}{\ensuremath{\mathbb{N}}}

\newcommand*{\exps}[1]{e^{#1}}
\newcommand*{\norDS}{\ensuremath{\mathit{gauss}}}
\newcommand*{\betaDS}{\ensuremath{\mathit{beta}}}

\newcommand*{\diracDS}{\ensuremath{\mathit{dirac}}}
\newcommand*{\berDS}{\ensuremath{\mathit{bern}}}

\newcommand*{\expDS}{\ensuremath{\mathit{exp}}}

\DeclareMathOperator*{\identity}{id}

\newcommand{\inv}{^{\textrm-1}}
\newcommand{\op}{^{\mathrm{op}}}
\newcommand*{\defeq}{\ensuremath{\,\smash{\stackrel{\text{\tiny def}}=}\,}}

\newcommand*{\SA}[1]{\Sigma_{#1}}

\newcommand*{\PM}{P} 
\newcommand*{\Monad}{T} 

\newcommand*{\ipjname}{\mathsf{d}}
\newcommand*{\icjname}{\mathsf{p}}
\newcommand*{\ixjname}{\mathsf{z}}

\newcommand*{\ipj}[3]{#1\mathrel{\vdash\!\!\!\!_\ipjname}#2 \colon #3}
\newcommand*{\icj}[3]{#1\mathrel{\vdash\!\!\!\!_\icjname}#2 \colon #3}
\newcommand*{\ixj}[3]{#1\mathrel{\vdash\!\!\!\!_\ixjname}#2 \colon #3}
\newcommand*{\vdashipj}{\mathrel{\vdash\!\!\!\!_\ipjname}}
\newcommand*{\vdashicj}{\mathrel{\vdash\!\!\!\!_\icjname}}

\newcommand*{\inferback}[2]{\infer{#2}{#1}}
\newcommand*{\inferbackusing}[3]{\infer[#3]{#2}{#1}}

\newcommand*{\db}[1]{\ensuremath{\llbracket #1 \rrbracket}}
\newcommand*{\dho}[1]{\ensuremath{\llparenthesis #1\rrparenthesis}}

\newcommand*{\inorm}{\ensuremath{\mathsf{norm}}}

\newcommand*{\itrue}{\ensuremath{\mathsf{true}}}
\newcommand*{\ifalse}{\ensuremath{\mathsf{false}}}
\newcommand*{\ilet}{\ensuremath{\mathsf{let}}}
\newcommand*{\iin}{\ensuremath{\mathsf{in}}}
\newcommand*{\icase}{\ensuremath{\mathsf{case}}}
\newcommand*{\iif}{\ensuremath{\mathsf{if}}}
\newcommand*{\ithen}{\ensuremath{\mathsf{then}}}
\newcommand*{\ielse}{\ensuremath{\mathsf{else}}}
\newcommand*{\iof}{\ensuremath{\mathsf{of}}}
\newcommand*{\ireturn}{\ensuremath{\mathsf{return}}}
\newcommand*{\ithunk}{\ensuremath{\mathsf{thunk}}}
\newcommand*{\iforce}{\ensuremath{\mathsf{force}}}
\newcommand*{\isample}{\ensuremath{\mathsf{sample}}}

\newcommand*{\iscore}{\ensuremath{\mathsf{score}}}

\newcommand*{\iboolty}{\ensuremath{\mathsf{bool}}}

\newcommand*{\density}{\ensuremath{\mathsf{D}}}

\newcommand*{\ctx}[1]{\mathcal C[#1]}
\newcommand*{\ectxc}[1]{\mathcal D[#1]}
\newcommand*{\ectxp}[1]{\mathcal E[#1]}
\newcommand*{\conf}[2]{\langle #1,#2\rangle}
\newcommand*{\ConZ}[1]{\mathrm{Con}_{\ixjname}(#1)}
\newcommand*{\ConZV}[1]{\mathrm{ConV}_{\ixjname}(#1)}
\newcommand*{\ConZN}[1]{\mathrm{ConN}_{\ixjname}(#1)}
\newcommand*{\ConP}[1]{\mathrm{Con}_{\ipjname}(#1)}
\newcommand*{\ConPV}[1]{\mathrm{ConV}_{\ipjname}(#1)}
\newcommand*{\ConPN}[1]{\mathrm{ConN}_{\ipjname}(#1)}
\newcommand*{\ConC}[1]{\mathrm{Con}_{\icjname}(#1)}
\newcommand*{\ConCV}[1]{\mathrm{ConV}_{\icjname}(#1)}
\newcommand*{\ConCN}[1]{\mathrm{ConN}_{\icjname}(#1)}
\newcommand*{\RedCI}[2]{\Prob{#1\longrightarrow (#2)}}
\newcommand*{\bigRedCI}[2]{\bigProb{#1\longrightarrow (#2)}}
\newcommand*{\RedC}[2]{\Prob{#1\longrightarrow #2}}
\newcommand*{\PrEval}[2]{\Prob{#1\Downarrow #2}}
\newcommand*{\PrEvaln}[3]{\Prob{#2\Downarrow^{#1} #3}}
\newcommand*{\RedCC}[3]{\Prob{#1\longrightarrow_{#2}#3}}
\newcommand*{\bigRedCC}[3]{\bigProb{#1\longrightarrow_{#2}#3}}

\newcommand*{\Prob}[1]{\mathrm{Pr}(#1)}
\newcommand*{\bigProb}[1]{\mathrm{Pr}\big(#1\big)}
\newcommand*{\Score}{\mathrm{Sc}}
\newcommand*{\indicator}[1]{\left[{#1}\right]}

\newcommand*{\bA}{\ensuremath{\mathbb{A}}}
\newcommand*{\bB}{\ensuremath{\mathbb{B}}}

\newcommand*{\bD}{\ensuremath{\mathbb{D}}}

\newcommand*{\iPM}{\mathsf{P}}
\newcommand*{\iMonad}{\mathsf{T}}

\newcommand{\D}{\ensuremath{\mathrm{d}}}

\newcommand{\ie}{\textit{i.e.}\xspace}
\newcommand{\eg}{\textit{e.g.}\xspace}

\newcommand{\Meas}{\ensuremath{\mathbf{Meas}}}
\newcommand{\Set}{\ensuremath{\mathbf{Set}}}

\begin{document}

\setlength{\pdfpageheight}{\paperheight}
\setlength{\pdfpagewidth}{\paperwidth}

\conferenceinfo{CONF 'yy}{Month d--d, 20yy, City, ST, Country}
\copyrightyear{20yy}
\copyrightdata{978-1-nnnn-nnnn-n/yy/mm}
\copyrightdoi{nnnnnnn.nnnnnnn}



\title{Semantics for probabilistic programming: higher-order functions, continuous distributions,
and soft constraints}

\authorinfo{Sam Staton\and Hongseok Yang \and Frank Wood} 
{University of Oxford}
{}
\authorinfo{Chris Heunen} 
{University of Edinburgh}
{}
\authorinfo{Ohad Kammar}
{University of Cambridge}
{}

\maketitle

\begin{abstract}

We study the semantic foundation of expressive probabilistic programming languages, that support higher-order functions, continuous distributions, and soft constraints 
(such as Anglican, Church, and Venture).
%
We define a metalanguage (an idealised version of Anglican) for probabilistic computation with the above features, 
develop both operational and denotational semantics,
and prove soundness, adequacy, and termination.
This involves measure theory, stochastic labelled transition systems, and functor categories, but admits intuitive computational readings, one of which views sampled random variables as dynamically allocated read-only variables. We apply our semantics 
to validate nontrivial equations underlying the correctness of certain compiler optimisations and inference algorithms such as sequential Monte Carlo simulation. 
The language enables defining probability distributions on higher-order functions, and we study their properties.


\end{abstract}



\section{Introduction}
\label{sec:intro}

Probabilistic programming is the idea to use programs to specify probabilistic models;
probabilistic programming languages blend programming constructs with probabilistic primitives.
This helps scientists express complicated models succinctly.
Moreover, such languages come with generic inference algorithms,
relieving the programmers of the nontrivial task of (algorithmically) answering queries about their probabilistic models.
This is useful in \eg~machine learning.

Several higher-order probabilistic programming languages have recently attracted a substantial user base.
Some languages (such as Infer.net~\cite{minka_software_2010}, PyMC~\cite{patilhuardfonnesbeck:pymc}, and Stan~\cite{stan_software_2014}) 
are less expressive but provide powerful inference algorithms, 
while others (such as Anglican~\cite{wood-aistats-2014}, Church~\cite{goodman_uai_2008}, 
and Venture~\cite{Mansinghka-venture14}) have less efficient inference algorithms but more expressive power.
We consider the more expressive languages, that support higher-order functions, continuous distributions, and soft constraints.
More precisely, we consider a programming language (\S\ref{sec:fo-lang}) with higher-order functions (\S\ref{sec:ho-lang}) as well as the following probabilistic primitives.
\begin{description}
  \item[Sampling] The command $\isample(t)$ draws a sample from a distribution described by $t$, which may range over the real numbers.
  \item[Soft constraints] The command $\iscore(t)$ puts a score $t$ (a positive real number) on the current execution trace. This is typically used to record that some particular datum was observed as being drawn from a particular distribution; the score describes how surprising the observation is.
  \item[Normalization] The command $\inorm(u)$ runs a simulation algorithm over the program fragment $u$. This takes the scores into account and returns a new, normalized probability distribution.
  The argument to $\isample$ might be a primitive distribution, or a distribution defined by normalizing another program.
  This is called a \emph{nested query}, by analogy with database programming. 
\end{description}

Here is a simple example of a program. We write 
$\norDS(\mu,\sigma)$ for the Gaussian probability distribution whose density function is
$\mathit{density}{\_}\norDS(a,(\mu,\sigma))=\frac 1{\sigma\sqrt {2\pi}}\exp(-\frac{(a-\mu)^2}{2\sigma^2})$. 
\begin{equation}\label{eqn:example}
\begin{array}{ll}
1&\inorm(\\
2&\quad \ilet\,x=\isample(\norDS(0.0,3.0))\,\iin\\
3&\quad \iscore(\mathit{density}{\_}\norDS (5.0,(x,1.0));\\
4&\quad \ireturn(x<4.5))
\end{array}\end{equation}
Line~2 samples~$x$ from a prior Gaussian distribution. The soft constraint on Line~3 expresses the likelihood of 
the observed data,~$5.0$, coming from a Gaussian given the prior~$x$. Line~4 says that what we are actually interested in 
is a boolean random variable over the sample space.
Line~1 calculates a posterior distribution for the return value, using the prior and the likelihood. In this example we can precisely calculate that the posterior distribution on $\{\itrue,\ifalse\}$ has $p(\itrue)=0.5$.

Languages like this currently lack formal exact semantics.
The aim of this paper is to provide just such a foundation as a basis for formal reasoning.
Most expressive probabilistic programming languages are explained in terms of their Monte Carlo simulation algorithms. The simplest such algorithm, using importance and rejection sampling, is the \textit{de facto} semantics against which other algorithms are `proved approximately correct'. Such `semantics' are hard to handle and extend.
\hide{SS: I think I've covered these points below.
  \item It validates nontrivial program transformations, such as those involving nested queries.
  For example, Subsection~\ref{subsec:montecarlo} below proves correctness of sequential Monte Carlo simulation.
  \item It resolves a variation of the Church-Turing thesis: despite its higher-order features, can the language only express measurable functions between real numbers? Theorem~\ref{thm:churchturing} below proves that this conjecture (raised in~\cite{parkpfenningthrun:sampling}) indeed holds.
  \item It aids language design by identifying useful constructs, concepts, and typing disciplines.
  Various advanced nonparametric Bayesian models seem to have structural explanations.
  Studying their semantic interpretations brings out the commonalities, as illustrated in Section~\ref{sec:ho-opsem} below.
\end{itemize}}

We provide two styles of semantics, operational and denotational.
For first-order probabilistic programs, the denotational semantics is straightforward:
types are interpreted as measurable spaces, and terms are interpreted as measurable 
functions~(\S\ref{sec:fo-densem}). Operational semantics is more complicated. For discrete distributions, an 
operational semantics might be a probabilistic transition system, but for 
continuous distributions, it must be a stochastic relation (labelled Markov process).
We resolve this by equipping the set of configurations with the structure of a measurable space~(\S\ref{sec:fo-opsem}).

The advantage to the operational semantics is that it is easily extended to higher-order programs~(\S\ref{sec:ho-opsem}).
Denotational semantics for higher-order programs poses a problem, because 
measurable spaces do not support the usual $\beta/\eta$ theory of functions:
they do not form a Cartesian closed category (indeed, $\R^\R$ does not exist as a measurable space~\cite{aumann:functionspaces}).
Earlier work dealt with this either by excluding higher-order functions or by considering only discrete distributions.
We resolve this by moving from the category of measurable spaces, where standard probability theory takes place, to a functor category based on it~(\S\ref{sec:ho-densem}). 
The former embeds in the latter, so we can still interpret first-order concepts. But the functor category does have well-behaved function spaces, so we can also interpret higher-order concepts. Moreover, by lifting the monad of probability distributions~\cite{giry:monad} to the functor category, we can also interpret continuous distributions. 
Finally, we can interpret observations by considering probability distributions with continuous density, irrespective of the categorical machinery~(\S\ref{sec:density}).

The denotational semantics is sound and adequate with respect to the operational semantics (\S\ref{sec:fo-opsem-sound},\ref{sec:ho-sound}), which means one can use the denotational model to directly check program equations while respecting computational issues.
For example:
\begin{itemize}
  \item we demonstrate a key program equation for sequential Monte Carlo simulation~(\S\ref{subsec:montecarlo});
  \item we show that every term of first-order type is equal to one without $\lambda$-abstractions or application, and hence is interpreted as a measurable function~(Proposition~\ref{prop:church-turing}).
\end{itemize}

\section{Preliminaries}
\label{sec:prelim}

We recall basic definitions and facts of measure theory.

\begin{definition}
  A \emph{$\sigma$-algebra} on a set $X$ is a family $\Sigma$ of subsets of $X$, called \emph{measurable (sub)sets}, which contains $X$ and is closed under complements and countable unions.
  A \emph{measurable space} is a set with a $\sigma$-algebra.

  A \emph{measure} on a measurable space $(X,\Sigma)$ is a function $p \colon \Sigma\to [0,\infty]$ such that 
$p(\emptyset)=0$ and 
  $p(\bigcup_{i\in \N}U_i)=\sum_{i\in\N} p(U_i)$ for each sequence $U_1,U_2,\ldots$ of disjoint measurable sets.
  A \emph{probability measure} or \emph{probability distribution} is a measure $p$ with $p(X)=1$. 
\end{definition}

In the paper we use a few important constructions for measurable spaces.
The first is to make a set $X$ into a measurable space by taking the full powerset of $X$ as $\Sigma$, 
yielding a \emph{discrete} measurable space. When $X$ is countable, a probability distribution on $(X,\Sigma)$ is determined 
by its values on singleton sets, that is, by specifying a function $p \colon X\to [0,1]$ such that $\sum_{x\in X}p(x)=1$. 

The second construction is to combine a collection of measurable spaces $(X_i,\Sigma_i)_{i \in I}$ by \emph{sum} or \emph{product}. The underlying sets are the disjoint union $\sum_{i \in I} X_i$ and product $\prod_{i \in I} X_i$ of sets.
The measurable sets in the sum are $\sum_{i \in I} U_i$ for $U_i \in \Sigma_i$.
The $\sigma$-algebra of the product is the smallest one containing all the subsets $\prod_{i \in I} U_i$ where $U_i \in \Sigma_i$ equals $X_i$ but for a single index $i$.

For the third, the real numbers form a measurable space $(\R,\Sigma_\R)$ under the smallest $\sigma$-algebra that contains the open intervals; the measurable sets are called \emph{Borel} sets. Restricting to any measurable subset gives a new measurable space, such as the space $\Rpz$ of nonnegative reals and the unit interval $[0,1]$.

The fourth construction is to make the set $\PM(X)$ of all probability measures on a measurable space $(X,\Sigma_X)$ into a measurable space, by letting $\Sigma_{\PM(X)}$ be the smallest $\sigma$-algebra containing the sets $\{p\in\PM(X) \mid p(U)\in V\}$ for all 
$U\in \Sigma_X$ and $V\in \Sigma_{[0,1]}$.  

\begin{definition}
  Let $(X,\Sigma_X)$, $(Y,\Sigma_Y)$ be measurable spaces. 
  A function $f\colon X\to Y$ is \emph{measurable} if $f\inv(U)\in\Sigma_X$ for $U \in \Sigma_Y$. 
\end{definition}

We can \emph{push forward} a measure along a measurable function:
if $p \colon \Sigma_X\to [0,1]$ is a probability measure on $(X,\Sigma_X)$ and $f\colon X\to Y$ is a measurable function, then 
$q(U)=p(f\inv(U))$ is a probability measure on $(Y,\Sigma_Y)$. 

\begin{definition}\label{def:srel}
  A \emph{stochastic relation} between measurable spaces $(X,\Sigma_X)$ and $(Y,\Sigma_Y)$ is a function $r \colon X\times \Sigma_Y\to [0,1]$ such that $r(x,-) \colon \Sigma_Y\to[0,1]$ is a probability distribution for all $x\in X$,
  and $r(-,V) \colon X\to[0,1]$ is measurable for all $V\in\Sigma_Y$.
\end{definition}

Giving a stochastic relation from $(X,\Sigma_X)$ to $(Y,\Sigma_Y)$ is equivalent to giving a measurable function $(X,\Sigma_X) \to (\PM(Y),\Sigma_{\PM(Y)})$.
Stochastic relations $r \colon X\times \Sigma_Y\to [0,1]$ and $s\colon Y\times \Sigma_Z\to [0,1]$ compose associatively to $(s\circ r)\colon X\times \Sigma_Z\to [0,1]$ via the formula 
\[
  (s\circ r)(x,W)=\int_Y s(y,W)\ r(x,\D y).
\]

Finally, for a predicate $\varphi$, we use the indicator expression $\indicator{\varphi}$ to denote $1$ if $\varphi$ holds, and $0$ otherwise.


\section{A first-order language}
\label{sec:fo-lang}

This section presents a first-order language for expressing Bayesian probabilistic models. The language forms a first-order core of a higher-order extension in Section~\ref{sec:ho-lang}, and provides a simpler setting to illustrate key ideas. 
The language includes infinitary type and term constructors, constant terms for all measurable functions between  measurable spaces, 
and constructs for specifying Bayesian probabilistic models, namely, operations for sampling distributions, scoring samples, and normalizing 
distributions based on scores.
This highly permissive and slightly unusual syntax is not meant to be a useful programming language itself.
Rather, its purpose is to serve as a semantic metalanguage to which a practical programming language compiles, and to provide a mathematical setting for studying high-level constructs for probabilistic computation.

\paragraph{Types} 
The language has types
\[
\bA,\bB \;::=\; 
\R
~|~ \iPM(\bA)
~|~ 1
~|~ \bA \times \bB
~|~ \sum_{i\in I}\bA_i 
\]
where $I$ ranges over countable sets. 
A type $\bA$ stands for a measurable space~$\db{\bA}$. For example, $\R$ denotes the measurable space of reals, $\iPM(\bA)$ is the space of probability measures on $\bA$, and $1$ is the (discrete) measurable space on the singleton set. 
The other type constructors correspond to products and sums of measurable spaces. Notice that countable sums are allowed, enabling us to express usual ground types in programming languages via standard encoding. For instance, the type for booleans is $1+1$, 
and that for natural numbers $\sum_{i\in \N}1$.

\paragraph{Terms}

We distinguish typing judgements: $\ipj\Gamma t \bA$ for deterministic terms, and $\icj\Gamma t \bA$ for probabilistic terms
(see also e.g.~\cite{lpt-cbv,parkpfenningthrun:sampling,ramseypfeffer:stochasticlambda}).
In both, $\bA$ is a type, and $\Gamma$ is a list of variable/type pairs. 
Variables stand for deterministic terms.

Intuitively, a probabilistic term $\icj\Gamma t \bA$ may have two
kinds of effects: during evaluation, $t$ may sample from a
probability distribution, and it may score the current execution trace
(typically according to likelihood of data). Evaluating a
deterministic term $\ipj\Gamma t \bA$ does not have any effects.

\paragraph{Sums and products}

The language includes variables, and standard constructors and destructors for sum and product types.
\[
  \begin{array}{c}
    \infer{\ipj{\Gamma,x \colon \bA,\Gamma'}x\bA}{}
    \qquad
    \inferbackusing{
      \ipj \Gamma t {\bA_i}
    }{
      \ipj \Gamma {(i,t)} {\sum_{i\in I}\bA_i}
    }{
    }
    \\
    \\
    \inferbackusing{
      \ipj\Gamma t {\sum_{i\in I}\bA_i}
      \quad 
      \left(\ixj {\Gamma,x\colon \bA_i} {u_i} {\bB}\right)_{i\in I}
    }{
      \ixj \Gamma {\icase\ t\ \iof\ \{(i,x)\Rightarrow u_i\}_{i \in I}} {\bB}
    }{
      (\ixjname\in\{\ipjname,\icjname\})
    }
    \\
    \\
    \inferbackusing{
    }{
        \ipj \Gamma {*} {1}
    }{
    }
    \quad
    \inferbackusing{
      \ipj{\Gamma}{t_j}{\bA_j} \ \text{for all $j \in \{0,1\}$}
    }{
      \ipj{\Gamma}{(t_0,t_1)}{\bA_0 \times \bA_1}
    }{
    }
    \quad
    \inferbackusing{
      \ipj \Gamma t {\bA_0 \times \bA_1} 
    }{
      \ipj \Gamma {\pi_j(t)} {\bA_j}
    }{
    }
  \end{array}
\]
In the rules for sums, $I$ may be infinite.
In the last rule, $j$ is $0$ or $1$. 
We use some standard syntactic sugar, such as $\ifalse$ and $\itrue$ for the injections in the type $\iboolty = 1+1$, and 
 $\iif$ for $\icase$ in that instance.

\paragraph{Sequencing}

We include the standard constructs (e.g.~\cite{lpt-cbv,moggi-monads}).
\[
  \begin{array}{c}
    \inferback{
     \ipj\Gamma t \bA
    }{
     \icj \Gamma {\ireturn(t)} {\bA}
    }
    \qquad\qquad
    \inferback{
      \icj\Gamma t \bA
      \quad\icj {\Gamma,x \colon \bA} u \bB
    }{
     \icj \Gamma {\ilet\ x=t\ \iin\ u} {\bB}
    }
  \end{array}
\]
Where $\bA=1$, we write $(t;u)$ for $\ilet\,x=t\,\iin\,u$.
\paragraph{Language-specific constructs}

The language has constant terms for all measurable functions.
\begin{equation}\label{eq:functionterms}
  \begin{aligned}
        \infer[\text{($f \colon \db{\bA} \to \db{\bB}$ measurable)}]{
                \ipj\Gamma {f(t)}{\bB}
        }{ 
                \ipj\Gamma t {\bA} 
        }
  \end{aligned}
\end{equation}
In particular, all the usual distributions are in the language, including the Dirac distribution $\diracDS(x)$ concentrated on outcome $x$, the Gaussian distribution $\norDS(\mu,\sigma)$ with mean $\mu$ and standard deviation $\sigma$, the Bernoulli distribution $\berDS(p)$ with success probability $p$, the exponential distribution $\expDS(r)$ with rate $r$, 
and the Beta distribution $\betaDS(\alpha,\beta)$ with parameters
$\alpha,\beta$.\footnote{Usually, $\norDS(0.0,0.0)$ is 
 undefined, but we just let $\norDS(0.0,0.0)=\norDS(0.0,1.0)$, and so on,
 to avoid worrying about this sort of error.}
For example, from the measurable functions $42.0 \colon 1\to\R$, $\exps{(-)} \colon \R\to \R$, $\norDS \colon \R\times \R\to\PM(\R)$ and ${<}\colon \R\times \R\to 1+1$ we can derive: 
\[
        \infer{
                \ipj\Gamma {42.0}{\R}
        }{ 
        }
\qquad
        \infer{
                \ipj\Gamma {\exps{t}}{\R}
        }{ 
                \ipj\Gamma {t} {\R} 
        }
\]
\[
        \infer{
                \ipj\Gamma {\norDS(\mu,\sigma)}{\iPM(\R)}
        }{ 
                \ipj\Gamma {\mu} {\R} 
                \quad
                \ipj\Gamma {\sigma} {\R} 
        }
\qquad
        \infer{
                \ipj\Gamma {t<u}{\iboolty}
        }{ 
                \ipj\Gamma t {\R} 
                \quad
                \ipj\Gamma u{\R} 
        }
\]

The following terms form the core of our language.
\[
        \infer{ 
                \icj\Gamma {\isample(t)} \bA
        }{
                \ipj\Gamma t {\iPM(\bA)}
        }
        \qquad\qquad
        \infer{ 
                \icj\Gamma {\iscore(t)} 1
        }{ 
                \ipj\Gamma t \R
        }
\]
The first term samples a value from a distribution $t$. The second multiplies the score of the current trace 
with $t$.
Since both of these terms express effects, they are typed under ${\vdashicj}$ instead 
of ${\vdashipj}$. 

The reader may think of the score of the current execution trace as a state, but it cannot be read, nor changed arbitrarily: 
it can only be multiplied by another score.
The argument $t$ in $\iscore(t)$ is usually the density of a
probability distribution at an observed data point. For instance, recall
the example~\eqref{eqn:example} from the Introduction:
\[\begin{array}{l@{}l}\inorm(&
\ilet\,x=\isample(\norDS(0.0,3.0))\,\iin\,\\&
\iscore(\mathit{density}{\_}\norDS (5.0,(x,1.0));\ireturn(x<4.5))
\end{array}\] An 
execution trace is scored higher 
when~$x$ is closer to the datum~$5$. 

When the argument $t$ in $\iscore(t)$ is $0$, this is called a 
\emph{hard constraint}, meaning `reject this trace'; otherwise it is
called a soft constraint. In the discrete setting, hard constraints are
more-or-less sufficient, but even then, soft constraints tend to be more efficient.

\paragraph{Normalization} 

Two representative tasks of Bayesian inference are to calculate the
so-called \emph{posterior distribution} and \emph{model
  evidence} from a prior distribution and likelihood.
Programs built from $\isample$ and $\iscore$ can be thought of as
setting up a prior and a likelihood. Consider the following program:
\[
        \begin{array}{@{}l@{}} 
                \ilet\ x=\isample(\berDS(0.25))\ \iin\ 
                \\
          \ilet\ y=(\iif\ x\ \ithen \ \ireturn\ 5.0\ \ielse\ \ireturn\
          2.0)\ \iin
               \\ \iscore(\mathit{density}\_\expDS(0.0,y));
                \ireturn(x)
        \end{array}
\]

Here the prior $y$ comes from the Bernoulli distribution, and the
likelihood concerns datum $0.0$ coming from an exponential
distribution with rate $y$.
Recall that $\mathit{density}\_\expDS(0.0,y)=y$.
So there are two execution traces, returning either $\itrue$, with probability $0.25$ and 
score 
$5.0$,
or $\ifalse$, with probability $0.75$ and score $2.0$. 

The  product of the prior and likelihood gives an unnormalized
posterior distribution on the return value:
$({\itrue\mapsto 0.25{\cdot} 5{=}1.25},$ ${\ifalse\mapsto 0.75{\cdot} 2{=}1.5})$. 
The normalizing constant is  the average score: 
$(0.25 {\cdot} 5 + 0.75 {\cdot} 2) = 2.75$, so the posterior is
$(\itrue\mapsto \frac{1.25}{2.75}{\approx}0.45,\ifalse\mapsto
\frac{1.5}{2.75}{\approx} 0.55)$. 
The average score is called the model evidence. It is a measure of how well the model encoded by the program matches the observation. 
%

Note that the sample $x=\itrue$ matches the datum better, so 
the probability of $\itrue$ goes up from $0.25$ to $0.45$ in the posterior.

In our language we have a term $\inorm(t)$ that will usually convert a probabilistic
term $t$ into a deterministic value, which is its posterior
distribution together with the model evidence. 
\[
        \infer{ 
                \ipj \Gamma {\inorm(t)} {(\R \times \iPM(\bA))\ +1+1}
        }{ 
                \icj \Gamma t {\bA}
        }
\]
If the model evidence is $0$ or $\infty$, the conversion fails, and
this is tracked by the `${+}1{+}1$'. 


\section{Denotational semantics}
\label{sec:fo-densem}
This section discusses the natural denotational semantics of the first-order language.
The basic idea can be traced back a long way (\eg \cite{kozen:probablistic}) 
but our treatment of $\iscore$ and $\inorm$ appear to be novel.
 As described, types $\bA$ are interpreted as
measurable spaces $\db{\bA}$. A 
context $\Gamma=(x_1 \colon \bA_1,\ldots,x_n \colon \bA_n)$ is
interpreted as the measurable space $\db\Gamma \defeq
\prod_{i=1}^n\db{\bA_i}$ of its valuations. 

\begin{itemize}
  \item Deterministic terms $\ipj\Gamma t \bA$ are interpreted as measurable functions $\db t \colon \db\Gamma\to\db\bA$, providing a result for each valuation of the context.
  \item Probabilistic terms $\icj\Gamma t \bA$ are interpreted as
    measurable functions $\db t \colon \db\Gamma\to\PM(\Rpz\times
    \db\bA)$, providing a probability measure on (score,result) pairs 
    for each valuation of the context. 
\end{itemize}
    Informally, if 
    $(\Omega,p\colon\Omega{\to}[0,1])$ is the prior sample space of the
    program (specified by $\isample$), 
    and $l\colon\Omega{\to}\Rpz$ is the likelihood (specified by $\iscore$),
    and $r\colon\Omega\to\db\bA$ is the return value, 
    then a measure in $\PM(\Rpz\times\db\bA)$ is found by 
    pushing forward $p$ along~$(l,r)$.
 

\paragraph{Sums and products}

The interpretation of deterministic terms follows the usual pattern of the internal language of a distributive category (\eg \cite{pitts-catlogic}). 
For instance, $\db{\ipj{\Gamma,x\colon \bA,\Gamma'}x \bA}(\gamma,a,\gamma')\defeq a$, 
and $\db{\ipj{\Gamma}{f(t)}\bA}(\gamma)\defeq f(\db t(\gamma))$ for measurable $f\colon\db \bA\to\db \bB$.
This interpretation is actually the same as the usual set-theoretic semantics of the calculus, as one can show by induction that the induced functions $\db\Gamma\to\db A$ are measurable. 

\paragraph{Sequencing}

For probabilistic terms, we proceed as follows. 
\[
  \db{\ireturn(t)}(\gamma)(U)\defeq [(1,\db t(\gamma))\in U],
\]
which denotes a Dirac distribution, and $\db{\ilet\,x=t\,\iin\,u}(\gamma)(V)$ is
\[
{        \int_{\Rpz \times \db\bA} 
        \Bigl(\db u(\gamma,x)\big(\big\{(s,b)\,\big|\, (r\cdot s,b)\,{\in}\,V\big\}\big)\Bigr)\,\Bigl(\db{t}(\gamma)(\D(r,x))\Bigr).}
\]
As we will explain shortly, these interpretations come from treating $P(\Rpz\times(-))$ as a commutative monad, which essentially means 
the following program equations hold.
\begin{align*}
&\db{\ilet\, x=\ireturn(x)\,\iin\,u}  =\db u \qquad
  \db{\ilet\, x=t\,\iin \,\ireturn(x)}  =\db t \\
&\begin{aligned}
  \db{\ilet\, y=(\ilet\,x=t\,\iin\,u) \,\iin\,v} & = \db{\ilet\, x=t\,\iin\,\ilet\,y=u \,\iin\,v} \\
  \db{\ilet\,x=t\,\iin\,\ilet\,y=u\,\iin\,(x,y)} & = \db{\ilet\,y=u\,\iin\,\ilet\,x=t\,\iin\,(x,y)}
\end{aligned}\end{align*}
The last equation justifies a useful program optimisation technique~\cite[\S5.5]{wood-aistats-2014}.
(The last two equations require assumptions about free variables of $u$,
$v$ and $t$, to avoid variable capture.)
\paragraph{Language-specific constructs} 

We use the monad:
\begin{align*}
  \db{\isample(t)}(\gamma)(U) & \defeq \db t(\gamma)(\{a~|~(1,a)\in U\}) \\
  \db{\iscore(t)}(\gamma)(U) & \defeq [(\max(\db t(\gamma),0),*)\in U]
\end{align*}
Here are some program equations to illustrate the semantics so far.
\begin{align*}
  \db{\iscore(7.0);\iscore(6.1)} & = \db{\iscore(42.7)} \\
  \left\llbracket\hspace*{-1.5mm}
  \begin{array}{l}
    \ilet\,x=\isample(\norDS(0.0,1.0))\\
    \iin\,\ireturn(x>0.0)
  \end{array}
  \hspace*{-1.5mm}\right\rrbracket\hspace*{-.5mm}
  &=\db{\isample(\berDS(0.5))} \\
  \left\llbracket\hspace*{-1.5mm}
  \begin{array}{l}
    \ilet\,x=\isample(\norDS(0.0,1.0))\\
    \iin\,\ireturn(x>x)
  \end{array}
  \hspace*{-1.5mm}\right\rrbracket\hspace*{-.5mm}
  &=\db{\ireturn(\ifalse)}
\end{align*}
 
\paragraph{Normalization}
We interpret $\inorm(t)$ as follows.
Each probability measure $p$ on $(\Rpz\times X)$ induces an
unnormalized posterior 
measure $\bar p$ on $X$: let
$\bar p(U)=  \int_{\Rpz \times U}r\;p(\D(r,x))$.
As long as the average score $\bar p(X)$ is not $0$ or $\infty$, we can normalize $\bar p$ to build a 
posterior probability measure on $X$. 
This construction forms a natural transformation
$\iota_X \colon \PM(\Rpz \times X)\to (\R \times \PM(X)) + 1 + 1$
\begin{equation}\label{eqn:iota}
\iota_X(p)\ \defeq \   \begin{cases}
    (1,*) & \text{if }\bar p(X)=0\\
    (2,*)&\text{if }\bar p (X)=\infty\\
    \Big(0,\big(\bar p (X),\,\lambda U.\; \frac{\bar p (U)}{\bar p(X)}\big)\Big)&\text{otherwise}
  \end{cases}
\end{equation}
Let $\db{\inorm(t)}(\gamma) \defeq \iota(\db t(\gamma))$.
Here are some examples:
\begin{align*}
&\left\llbracket\hspace*{-1.5mm}
\begin{array}{l@{}l}\inorm(&
\ilet\,x=\isample(\norDS(0.0,3.0))\,\iin\,\\&
\iscore(\mathit{density}{\_}\norDS (5.0,(x,1.0));\ireturn(x<4.5))
\end{array}
\right\rrbracket\hspace*{-1mm}
  \\&\qquad = \textstyle{\big(0,(0.949,\berDS(0.5))\big) }
  \\
        &\left\llbracket\hspace*{-1.5mm}\begin{array}{l@{}l@{}} 
                \inorm\big(&\ilet\ x=\isample(\berDS(0.25))\ \iin\ 
                \\
                                          &(\iif\ x\ \ithen \ \iscore(5.0)\ \ielse\ \iscore(
          2.0));
                \ireturn(x)\big)
        \end{array}\right\rrbracket\hspace*{-1mm}
  \\&\qquad = \textstyle{\big(0,(2.75,\berDS(\frac 5{11}))\big) }
  \\
  &\llbracket
  \inorm\big(\ilet\,x =
    \isample(\expDS(1.0))\,\iin\,\iscore(\exps{x})\big)\rrbracket
    = (2,*) 
  \\
  &\left\llbracket\hspace*{-1.5mm}
  \begin{array}{l@{}l}
    \inorm\big(& \ilet\,x=\isample(\betaDS(1,3))\\ 
           & \iin\,\iscore(x); \ireturn(x)\big)
  \end{array}
  \hspace*{-1.5mm}\right\rrbracket\hspace*{-.5mm}
  =\hspace*{-.5mm}
  \left\llbracket\hspace*{-1.5mm}
  \begin{array}{l}
          \inorm\big(
    \iscore(\frac{1}{1+3}); \\ 
          \ \ \isample(\betaDS(2,3))\big)
  \end{array}
  \hspace*{-1.5mm}\right\rrbracket
\end{align*}
The third equation shows how infinite model evidence errors can arise when working with infinite distributions.  In the last equation, the parameter $x$ of $\iscore(x)$ represents 
the probability of $\itrue$ under $\berDS(x)$. The equation expresses
the so called conjugate-prior relationship between
Beta and Bernoulli distributions, which has been
used to optimise probabilistic programs~\cite{yang_DALI2015}.

\paragraph{Monads} The interpretation of $\ilet$ and $\ireturn$ given above arises from 
the fact that $\PM\big(\Rpz\times (-)\big)$ is a commutative monad on the category of measurable spaces and measurable functions (see also~\cite[\S2.3.1]{doberkat-book}, \cite[\S6]{scibor:haskell}).
Recall that a commutative monad $(\Monad,\eta,\mu,\sigma)$ in general comprises an endofunctor $\Monad$ together with natural transformations $\eta_X \colon X\to\Monad(X)$, $\mu_X \colon \Monad(\Monad(X))\to\Monad(X)$, 
$\sigma_{X,Y} \colon X \times \Monad(Y)\to\Monad(X\times Y)$ satisfying some laws~\cite{kock:commutative}.
Using this structure we interpret $\ireturn$ and $\ilet$ following Moggi~\cite{moggi-monads}:
\begin{align*}
  \db{\icj\Gamma{\ireturn(t)}\bA}(\gamma) & \defeq {\eta_{\db\bA}\big(\db t(\gamma)\big)} \\
  \db{\icj{\Gamma}{\ilet\,x=t\,\iin\,u}\bB}(\gamma) 
  & \defeq \mu_{\db\bB}\big(\Monad(\db u)\big(\sigma_{\db\Gamma,\db\bA}(\gamma,\db t(\gamma))\big)\big)
\end{align*}
Concretely, we make $\PM(\Rpz\times (-))$ into a monad by combining the standard commutative monad structure~\cite{giry:monad} on~$\PM$ and the commutative monoid $(\Rpz,\cdot,1)$ with the monoid monad transformer.
\hide{
$(T^M,\eta^M,\mu^M,\sigma^M)$
with $T^M=T(M\times(-))$.
 The space $\Rpz$ forms a monoid under multiplication, so it remains for us to give 
a strong monad structure on $\PM$. First, $\PM$ forms a 
functor: for $f:X\to Y$, let $(\PM(f))(p)$ be the pushforward measure on~$Y$. 
The monad structure is given by 
\begin{align*}
&\eta_X(x)(U)\defeq[x\in U]
&&
\mu_X(p)(U)\defeq\int_q q(U)\,p(\D q)
\\[-2mm]
&\text{where\ }\\[-4mm]
&[x\in U]=\begin{cases}1&(x)\in U\\0&\text{otherwise}\end{cases}
&&\sigma_{X,Y}(x,p)(U)\defeq p\{y~|~(x,y)\in U\}
\end{align*}
}

(We record a subtle point about $\iota$~\eqref{eqn:iota}. It is \emph{not} a monad morphism,
and we we cannot form a commutative monad of \emph{all} measures because the Fubini property
does not hold in general.)

\subsection{Sequential Monte Carlo simulation}
\label{subsec:montecarlo}

The program equations above justify some simple program transformations. 
For a more sophisticated one, consider \emph{sequential Monte Carlo simulation}~\cite{doucetfreitasgordon:montecarlo}.
A key idea of its application to probabilistic programs is: `Whenever there is a $\iscore$, it is good to renormalize and resample'.
This increases efficiency by avoiding too many program executions with low scores~\cite[Algorithm~1]{paige-icml-2014}.

The denotational semantics justifies the soundness of this transformation. 
For a term with top-level $\iscore$, \ie a term of the form $(\ilet \,x=t\,\iin\,(\iscore(u);v))$ where $u$ and $v$ may have $x$ free:
%
%
\begin{align*}
&\db{\inorm(\ilet \,x=t\,\iin\,(\iscore(u);v))}
\\
&\begin{array}{@{}r@{}ll@{}}
=\ \db{\inorm(\icase\,&(\inorm(\ilet\,x=t\,\iin\,\iscore(u);\ireturn(x)))\,\iof\,
&\tiny 1
\\&(0,(e,d))\Rightarrow \iscore(e);\ilet\,x=\isample(d)\,\iin\,v
&\tiny 2
\\~|~&
(1,\ast)\Rightarrow \iscore(0) ;\ireturn(w)
&\tiny 3
\\~|~&
(2,\ast)\Rightarrow \ilet \,x=t\,\iin\,(\iscore(u);v))}
&\tiny 4
\end{array}\end{align*}
Let us explain the right hand side.
Line 1 renormalizes the program after the $\iscore$, and in non-exceptional execution 
returns the model evidence $e$ and a new normalized distribution $d$.
Line 2 immediately records the evidence $e$ as a score, and then resamples $d$, using the resampled value in the continuation $v$. 
Line 3 propagates the error of $0$: $w$ is a deterministic term of the right type whose choice does not matter.
Finally, line 4 detects an infinite evidence error, and undoes the transformation. This error does
not arise in most applications of sequential Monte Carlo simulation.

%
%
%


\section{Operational semantics}
\label{sec:fo-opsem}

In this section we develop an operational semantics for the first-order language.
There are several reasons to consider this, even though the denotational semantics is arguably straightforward.
First, extension to higher-order functions is easier in operational semantics than in denotational semantics.
Second, operational semantics conveys computational intuitions that are obscured in the denotational semantics.
We expect these computational intuitions to play an important role in studying approximate techniques for performing posterior inference, such as sequential Monte Carlo, in the future.

Sampling from probability distributions complicates operational semantics. 
Sampling from a discrete distribution can immediately affect control flow. 
For example, in the term
\[
        \begin{array}{@{}l@{}} 
                \ilet\,x\,{=}\,\isample(\berDS(0.5))\ \iin\
                \iif\,x\,\ithen\,\ireturn(1.1)\,\ielse\,\ireturn(8.0) 
        \end{array}
\]
the conditional depends on the result of sampling the Bernoulli distribution. The result is $1.1$ with probability $0.5$ (\textit{cf.}~\cite[\S2.3]{borgstrometal:bayesiantransformer}).

Sampling a distribution on $\R$ introduces another complication. 
Informally, there is a transition
\[
  \isample(\norDS(0.0,1.0))\longrightarrow \ireturn(r)
\] 
for every real $r$, but any single transition has zero probability.
We can assign non-zero probabilities to sets of transitions; informally:
\[
  \Pr\Bigl(\isample(\norDS(0.0,1.0))\longrightarrow \{\ireturn(r)~|~r\leq 0\}\Bigr) = 0.5\text.
\]
To make this precise we need a $\sigma$-algebra on the set of terms, which can be done using \emph{configurations} rather than individual terms. A configuration is a \emph{closure}: a pair $\conf{t}{\gamma}$ of a term $t$ with free variables and an environment $\gamma$ giving values for those variables as elements of a measurable space.  (See also \cite[\S3]{HurNRS15}, 
\cite[\S3]{borgstrometal:op-sem}.)

Sampling a distribution $p$ on $\R+\R$ exhibits both complications:
\begin{equation}
\begin{array}{r@{}l}
\ilet\,x=\isample(p)\,\iin\,\icase\,x\,\iof\,&(0,r)\Rightarrow \ireturn(r+1.0)
\\[0.5ex] |&(1,r)\Rightarrow \ireturn(r-1.0)
\end{array}
\label{eqn:example-sampleR+R}
\end{equation}
The control flow in the $\icase$ distinction depends on which summand is sampled, but there is potentially a continuous distribution over the return values.
We handle this by instantiating the choice of summand in the syntax, but keeping the 
value of the summand in the environment, so that expression~\eqref{eqn:example-sampleR+R} can make a step to the closure
\begin{equation*}
  \conf{\raisebox{-\baselineskip}{$\begin{array}{l}
    \ilet\ x = \ireturn(0,y)\,\iin \\
    \icase\,x\,\iof\,(0,r)\Rightarrow \ireturn (r+1.0) \\
    \phantom{\icase\,x\iof}|(1,r)\Rightarrow \ireturn (r-1.0)
  \end{array}$}}{y \mapsto 42.0}.
\end{equation*}

\newcommand{\onorm}{\nu}
A type is \emph{indecomposable} if it has the form $\mathbb R$ or $\iPM(\mathbb A)$, and a context $\Gamma$ is \emph{canonical} if it only involves indecomposable types.

\paragraph{Configurations}
Let $\ixjname\in\{\ipjname,\icjname\}$. A \emph{$\ixjname$-configuration} of type $\bA$ is a triple $\conf {\Gamma,t}\gamma$ comprising
a canonical context $\Gamma$, 
a term $\ixj \Gamma t \bA$, and an element $\gamma$ of the measurable space $\db{\Gamma}$.
We identify contexts that merely rename variables, such as
\begin{align*}
  &\conf{(x\colon \R,y\colon\iPM(\R)),f(x,y)}{(x\mapsto 42.0,y\mapsto \norDS(0.0,1.0))} \\
  \approx
  &\conf{(u\colon\R,v\colon\iPM(\R)),f(u,v)}{(u\mapsto 42.0,v\mapsto \norDS(0.0,1.0))}.
\end{align*}
We call d-configurations \emph{deterministic configurations}, and p-config\-urations \emph{probabilistic configurations}; they differ only in typing.
We will abbreviate configurations to $\conf{t}{\gamma}$ when $\Gamma$ is obvious.

\emph{Values} $v$ in a canonical context $\Gamma$ are well-typed deterministic terms of the form
\begin{equation}
        \label{eqn:fo-value}
        v,w \;::=\; x_i~\mid~*~\mid~(v,w)~\mid~(i,v)
\end{equation}
where $x_i$ is a variable in $\Gamma$.
Similarly, a probabilistic term $t$ in context $\Gamma$ is called \emph{probabilistic value} or \emph{p-value} if $t \equiv \ireturn(v_0)$ for some value $v_0$.  
Remember from Section~\ref{sec:fo-densem} that the denotational semantics of values is simple and straightforward.

Write $\ConP \bA$ and $\ConC \bA$ for the sets of deterministic and probabilistic configurations of type $\bA$, and make them into measurable spaces by declaring $U\subseteq \ConZ \bA$ to be measurable if the set $\{\gamma\in\db \Gamma ~|~\conf t \gamma\in U\}$ is measurable for all judgements $\ixj\Gamma t \bA$.
\begin{equation}\label{eq:configurationspace}
        \ConZ \bA  =  
        \sum_{\substack{
                (\Gamma,t)\\ 
                \text{$\Gamma$ canonical,}\\ 
                \text{$\Gamma \vdash_{\!\!\mathsf{z}}\;t : \bA$}}}
        \db{\Gamma}
\end{equation}

Further partition $\ConZ \bA$ into $\ConZV \bA$ and $\ConZN \bA$ based on whether
a term in a configuration is a value or not:
\begin{align*}
        \ConPV \bA & 
        {}= \{\langle \Gamma, t,\gamma \rangle \in \ConP \bA \,\mid\,\mbox{$t$ is a value}\} 
        \\
        \ConPN \bA & 
        {}= \{\langle \Gamma, t,\gamma \rangle \in \ConP \bA \,\mid\,\mbox{$t$ is not a value}\} 
        \\
        \ConCV \bA & 
        {}= \{\langle \Gamma, t,\gamma \rangle \in \ConC \bA \,\mid\,\mbox{$t$ is a p-value}\} 
        \\
        \ConCN \bA & 
        {}= \{\langle \Gamma, t,\gamma \rangle \in \ConC \bA \,\mid\,\mbox{$t$ is not a p-value}\} 
\end{align*}
Particularly well-behaved values are the \emph{ordered values} $\ipj {\Gamma} v {\mathbb A}$, where each variable appears exactly once, and in the same order as in $\Gamma$.

\begin{lemma}\label{lemma:ultimate-patterns}
  Consider a canonical context $\Gamma$, a type $\bA$, an ordered value $\ipj {\Gamma} v {\mathbb A}$, and the induced measurable function 
  \[
        \db{v}: \db{\Gamma} \to \db{\bA}.
  \]
  The collection of all such functions for given $\mathbb A$ is countable, and forms a coproduct diagram. 
\end{lemma} 
\begin{proof}
  By induction on the structure of types. The key fact is that every type is a sum of products of indecomposable ones, because the category of measurable spaces is distributive, \ie the canonical map $\sum_{i\in I}(\mathbb A\times \mathbb B_i) \to \mathbb A\times \sum_{i\in I}\mathbb B_i$ is an isomorphism.
\end{proof}

For example, ${\bA=(\R\times \iboolty)+(\R\times \R)}$ has 3~ordered values,
 first $(\ipj{x\colon\R}{(0,(x,\itrue))}\bA)$, second $(\ipj{x\colon\R}{(0,(x,\ifalse))}\bA)$, and third $(\ipj{x\colon\R,y\colon\R}{(1,(x,y))}\bA)$, inducing a canonical measurable isomorphism $\R + \R + \R\times \R\cong\db\bA$.

\paragraph{Evaluation contexts}

We distinguish three kinds of \emph{evaluation contexts}: $\ctx{-}$ is a context for a deterministic term with a hole for deterministic terms; $\ectxc-$ and $\ectxp-$ are contexts for probabilistic terms, the former with a hole for probabilistic terms, the latter with a hole for deterministic terms.
\begin{equation}\label{eqn:fo-ctx}
  \begin{aligned}
        \ctx{-} & \;::=\; 
        (-)
        ~|~ \pi_j\, \ctx{-} 
        ~|~ (\ctx{-},t)
        ~|~ (v, \ctx{-})
        ~|~ (i,\,\ctx{-})
        \\
        & \phantom{\;::=\;}
        ~|~\icase\ \ctx{-}\ \iof\ \{(i,x)\Rightarrow t_i\}_{i \in I}
        ~|~ f(\ctx{-})
        \\[1ex]
        \ectxc{-} & \;::=\; 
        (-)~|~
        \ilet~{x=\ectxc-}~\iin~t 
        \\[1ex]
        \ectxp{-} & \;::=\; 
        \ectxc {\ireturn[-]}
        ~|~ \ectxc{\isample[-]}
        ~|~ \ectxc{\iscore[-]}
                    \\
        & \phantom{\;::=\;}
        ~|~ \icase\ \ectxc{-}\ \iof\ \{(i,x)\Rightarrow t_i\}_{i \in I}
  \end{aligned}
\end{equation}
where $t$, $t_i$ are general terms and $v$ is a value. 
\hide{A \emph{redex} is a term that has one of the following forms:
\begin{align*}
        &
        \pi_j(v,w),
        \ \ \ \icase\ (i',v)\ \iof\ \{(i,x)\Rightarrow t_i\}_{i \in I},
        \ \ \ f(v),
\ \         \ \inorm(t),
        \\
        & 
        \ilet \,{x=\ireturn(v)}\,\iin\, t,
\ \ \         \iscore(v),
\ \ \         \isample(v).
\end{align*}
Redexes on the first line are \emph{deterministic redexes}, those on the last line \emph{probabilistic redexes}.}

\subsection{Reduction}

Using the tools developed so far, we will define a measurable function
for describing the reduction of d-configurations, and a stochastic
relation for describing reduction of p-configurations:
\begin{align*}
&        {\longrightarrow} \;:\; \ConPN{\bA} \to \ConP{\bA},
\\
&        {\longrightarrow} \;:\;
        \ConCN{\bA} \times \Sigma_{\Rpz \times \ConC{\bA}} \to [0,1],
\end{align*}
parameterized by a family of measurable `normalization' functions
\begin{equation}
\label{eqn:onorm}
        \onorm_{\bA} \;:\;\ConC{\bA} \to \bigl(\Rpz \times P(\db\bA)\bigr) + 1 + 1
\end{equation}
indexed by types $\bA$.

\paragraph{Reduction of deterministic terms}

Define a type-indexed family of relations ${\longrightarrow}\subseteq \ConPN \bA\times \ConP \bA$
as the least one that is closed under the following rules.
\begin{align*}
        & 
        \conf {\Gamma, \pi_j(v_0,v_1)}\gamma \longrightarrow \conf{\Gamma,v_j}{\gamma}
        \\[1ex]
        &
        \conf {\Gamma,\icase\ (i',v)\ \iof\ \{(i,x)\Rightarrow t_i\}_{i \in I}}\gamma 
        \longrightarrow \conf{\Gamma,t_{i'}[v/x]}{\gamma}
        \\[1ex] 
        & 
        \begin{array}{@{}l@{}} 
                \conf {\Gamma,f(w)}\gamma \longrightarrow \conf{(\Gamma,\Gamma'),v}{(\gamma,\gamma')} 
                \\[0.5ex]
                \ (\text{$w$  a value} \wedge \text{$\ipj{\Gamma'}v\bA$ an ordered value} 
                \wedge f(\db{w}(\gamma))\,{=}\,\db{v}(\gamma'))
        \end{array}
        \\[1ex] 
        &
        \begin{array}{@{}l@{}} 
                \conf {\Gamma,\inorm(t)}\gamma \longrightarrow \conf{(\Gamma,x{:}\R,y{:}\iPM(\bB)),(0,(x,y))}{\gamma[x{\mapsto} r,y{\mapsto} p]} 
                \\[0.5ex] 
                \ (\bA \,{=}\,  (\R {\times} \iPM(\bB)+1+1) \wedge
                \onorm_{\bB}(\conf{\Gamma,t}\gamma) \,{=}\, (0,(r,p)) \wedge
                        x,y \,{\not\in}\, \Gamma)
        \end{array}
        \\[1ex]
        &
        \begin{array}{@{}l@{}} 
                \conf {\Gamma,\inorm(t)}\gamma \longrightarrow \conf{\Gamma,(i,*)}{\gamma}
                \\[0.5ex] 
                \qquad (\bA \,{=}\, (\R\times \iPM(\bB)+1+1) \wedge \onorm_\bB(\conf{\Gamma,t}\gamma) \,{=}\, (i,*),\ 
          i\in\{1,2\})
        \end{array}
\hide{        \\[1ex]
        &
        \begin{array}{@{}l@{}} 
                \conf {\Gamma,\inorm(t)}\gamma \longrightarrow \conf{\Gamma,(2,*)}{\gamma} 
                \\[0.5ex] 
                \qquad (\bA \,{=}\, \R\times \iPM(\bB)+1+1 \wedge \onorm_\bB(\conf{\Gamma,t}\gamma) \,{=}\, (2,*))
        \end{array}}
        \\[1ex]
        & 
        \infer[\qquad(\text{$\ctx{-}$ is not $(-)$})]{ 
                \conf {\Gamma, \ctx t }\gamma \longrightarrow \conf {\Gamma',\ctx {t'}}{\gamma'}
        }{ 
                {\conf {\Gamma, t} \gamma} \longrightarrow \conf {\Gamma',t'} {\gamma'}
        }
\end{align*}
The rule for $f(w)$ keeps the original context $\Gamma$ and the closure $\gamma$ because they might be used in the continuation,
even though they are not used in $v$. The rules obey the following invariant.

\begin{lemma}\label{lemma:type-pres}
  If $\conf{\Gamma,t}\gamma\longrightarrow \conf{\Gamma',t'}{\gamma'}$, then $\Gamma'=(\Gamma,\Gamma'')$ and $\gamma'=(\gamma,\gamma'')$ for some $\Gamma''$ and $\gamma''\in\db{\Gamma''}$.
\end{lemma}
\begin{proof}By induction on the structure of derivations.\end{proof}
\noindent This lemma allows us to confirm that our specification of a relation ${\longrightarrow}\subseteq\ConPN\bA\times \ConP\bA$ 
is well-formed (`type preservation').

\begin{proposition}
  The induced relation is a measurable function.
\end{proposition}
\begin{proof}
  There are three things to show: 
  that the relation is entire (`progress');
  that the relation is single-valued (`determinacy');
  and that the induced function is measurable. 
  All three are shown by induction on the structure of terms.
  The case of application of measurable functions crucially uses Lemma~\ref{lemma:ultimate-patterns}.
\end{proof}

\paragraph{Reduction of probabilistic terms}

Next, we define the stochastic relation $\longrightarrow$ for probabilistic terms, combining two standard approaches:
for indecomposable types, which are uncountable, use labelled Markov processes, \ie give a distribution on the measurable set of resulting configurations; 
for decomposable types (sums, products \textit{etc.}), probabilistic branching is discrete and so a transition system labelled by probabilities suffices.

\begin{proposition}\label{prop:srel-comb}
  Let $(X_i)_{i\in I}$ be an indexed family of measurable spaces.
  Suppose we are given:
  \begin{itemize}
    \item a function $q\colon I\to[0,1]$ that is only nonzero on a countable subset $I_0 \subseteq I$,
      and such that $\sum_{i \in I_0} q(i)=1$;
    \item a probability measure $q_i$ on $X_i$ for each $i\in I_0$.
  \end{itemize}
  This determines a probability measure $p$ on $\sum_{i\in I}X_i$ by
  \[
    p(U)=\textstyle{\sum_{i\in I}q(i)\,q_i(\{a \mid (i,a)\in U\})}
  \]
  for $U$ a measurable subset of $\sum_{i\in I}X_i$,
  %
  %
\end{proposition}

We will use three entities to define the desired stochastic relation 
${\longrightarrow} \;:\; \ConCN{\bA} \times \SA{\Rpz \times \ConC{\bA}} \to [0,1]$.
\begin{enumerate}
  \item A countably supported probability distribution on the set $\{(\Gamma,t)~|~\icj\Gamma t \bA\}$ for each $C \in \ConCN{\bA}$. 
    We write $\RedCI C{\Gamma,t}$ for the probability of $(\Gamma,t)$. 
  \item A probability measure on the space $\db\Gamma$ for each $C \in \ConCN{\bA}$ and $(\Gamma,t)$ with $\RedCI C{\Gamma,t} \neq 0$. Write $\RedCC C{\Gamma,t}U$ for the probability of a measurable subset $U\subseteq \db \Gamma$.
  \item A measurable function $\Score \colon \ConCN{\bA} \to \Rpz$, representing the score of the one-step transition relation.
    %
    %
    (For one-step transitions, the score is actually deterministic.)
\end{enumerate}
These three entities are defined by induction on the structure of the syntax of $\bA$-typed p-configurations in Figure~\ref{fig:entities}. 
We combine them to define a stochastic relation as follows.
\begin{figure}[!b]
\hrule
  \begin{align*}
        &
        \RedCI{\conf{\Gamma,\ectxp t}{\gamma}}{\Gamma',\ectxp{t'}} \defeq
        \indicator{\conf{\Gamma,t}\gamma\longrightarrow \conf{\Gamma',t'}{\gamma'}}
        \\[1ex]
        & 
        \RedCI{\conf{\Gamma, \ectxc t}\gamma}
          {\Gamma', \ectxc{t'}}
        \defeq {\RedCI{\conf{\Gamma,t}\gamma}{\Gamma',t'}}
        \\[1ex]
        & 
        \RedCI{\conf{\Gamma,\ilet \,{x=\ireturn(v)}\,\iin\, t}\gamma}
          {\Gamma,t[v/x]} \defeq 1
        \\[1ex]
        & 
        \RedCI{\conf{\Gamma,\icase\,(j,v)\,\iof\,\{(i,x)\Rightarrow t_i\}_{i\in I}}\gamma}
        {\Gamma,t_{j}[v/x]} \defeq 1
        \\[1ex]
        &
        \RedCI{\conf{\Gamma,\iscore(v)}\gamma}{\Gamma,\ireturn(*)} \defeq 1
        \\[1ex]
        &
        \begin{array}{@{}l@{}}
        \RedCI{\conf{\Gamma,\isample(v)}\gamma}
          {(\Gamma,\Gamma'),\ireturn(v')}
          \\[0.5ex]
          \ \quad {} \defeq
          \begin{cases}
          \db{v}(\gamma)(\{\db{v'}(\gamma') \,\mid\,\gamma' \in \db{\Gamma'}\}) 
          &\parbox{3cm}{if $\ipj{\Gamma'}{v'}{\bA}$ ordered}
        \\[1ex]
        0&\text{otherwise}
        \end{cases}
        \end{array}
\end{align*}
\begin{align*} 
        &
        \begin{array}{@{}l@{}} 
                \RedCC{\conf{\Gamma,\ectxp t}\gamma} {{(\Gamma',\ectxp{t'})}} U 
                \\[0.5ex] 
                \ \quad
                {}\defeq \indicator{\conf{\Gamma,t}\gamma\longrightarrow \conf{\Gamma',t'}{\gamma'} \wedge \gamma'\,{\in}\, U} 
        \end{array}
        \\[1ex]
        &
        \RedCC{\conf{\Gamma, \ectxc t}\gamma}{{(\Gamma',\ectxc{t'})}}U
        \defeq \RedCC{\conf{\Gamma,t}\gamma}{(\Gamma',t')} U
        \\[1ex]
        &
        \RedCC{\conf{\Gamma,\ilet \,{x=\ireturn(v)}\,\iin\, t}\gamma}{(\Gamma,t[v/x])} U
         \defeq \indicator{\gamma \in U}
        \\[1ex]
        & \RedCC{\conf{\Gamma,\icase\,(j,v)\,\iof\,\{(i,x)\,{\Rightarrow}\, t_i\}_{i\in I}}\gamma\!}{\Gamma,t_j[v/x]}{\!U}
                \defeq \indicator{\gamma\,{\in}\, U}
        \\[1ex]
        \displaybreak[0]
        &
        \RedCC{\conf{\Gamma,\iscore(v)}\gamma}{(\Gamma,\ireturn(*))} U
        \defeq\indicator{\gamma \in U}
        \\[1ex]
        &
        \RedCC{\conf{\Gamma,\isample(v)}\gamma}
          {((\Gamma,\Gamma'),\ireturn(v'))} U 
          \\&\ \quad\defeq\frac{
             \db{v}(\gamma) (\{\db{v'}(\gamma') \,\mid\,\gamma' \in \db{\Gamma'} \wedge (\gamma,\gamma') \in U\})}
             {\db{v}(\gamma) (\{\db{v'}(\gamma') \,\mid\,\gamma' \in \db{\Gamma'}\})}
        \displaybreak[0]
\\[2ex]        &
        \begin{array}{@{}ll}
          \Score(\conf{\Gamma,\ectxp t}\gamma) \defeq 1
          &
            \Score(\conf{\Gamma, \ectxc t}\gamma) \defeq {\Score(\conf{\Gamma,t}\gamma)}
          \\[2ex]
         \Score(\conf{\Gamma,\isample(v)}\gamma) \defeq 1
          &
        \end{array}
        \\[1ex]
        &
          \Score(\conf{\Gamma,\iscore(v)}\gamma) \defeq \max(\db{v}(\gamma),0)
          \\[1ex]&
        \Score(\conf{\Gamma,\ilet \,{x=\ireturn(v)}\,\iin\, t}\gamma) \defeq 1
        \\[1ex]
        &
        \Score(\conf{\Gamma,\icase\,(j,v)\,\iof\,\{(i,x)\Rightarrow t_i\}_{i\in I}}\gamma) \defeq 1
  \end{align*}
  \vspace{-2ex}
  \caption{Entities used to define reduction of probabilistic terms\label{fig:entities}}
\end{figure}

\begin{proposition}
  The map $\ConCN\bA\times \Sigma_{\Rpz \times \ConC{A}} \to [0,1]$ that sends $(C,U)$ to $\RedC C U$, defined as
  \[
    \sum_{(\Gamma,t)} \bigRedCI{C}{\Gamma,t} 
        \bigRedCC{C}{\Gamma,t}{\{\gamma \mid (\Score(C),\conf{\Gamma,t}\gamma)\,{\in}\, U\}},
  \]
  is a stochastic relation.
\end{proposition}
\begin{proof}
  For each p-configuration $C = \conf{\_,t}{\_}$, use induction on $t$ to see that the probability distribution $\bigRedCI{C}{-}$ on pairs $(\Gamma',t')$ and the distribution $\RedCC{C}{(-)}{(-)}$ indexed by such pairs satisfy the conditions in Proposition~\ref{prop:srel-comb}.
  It follows that the partial evaluation $\RedC{C}{(-)}$ of the function in the statement is a probability measure, so it suffices to establish measurability of the other partial evaluation $\RedC{(-)}{U}$. 
  Recall that $\ConCN\bA$ is defined in terms of the sum of measurable spaces, and that all p-configurations in each summand have the same term. 
  Finally, use induction on the term shared by all p-configurations in the summand to see that the restriction of $\RedC{(-)}{U}$ to each summand is measurable.
\end{proof}

\subsection{Termination}
\label{subsec:fo-termination}

To see that the reduction process terminates, we first define the transitive closure. 
This is subtle, as sampling can introduce countably infinite branching; although each 
branch will terminate, the expected number of steps to termination can be infinite.

We use the deterministic transition relation to define an evaluation relation ${\Downarrow}\subseteq \ConP \bA\times\ConPV\bA$, by setting $C\Downarrow D$ if ${\exists n.\,C\Downarrow^n D}$, where
\[
\inferbackusing{}
{C\Downarrow^0 C}
{(C\in\ConPV\bA)}
\qquad
\inferback
{D\Downarrow^n E \qquad C\longrightarrow D}
{C\Downarrow^{n+1} E}
\]
To define evaluation for probabilistic configurations, we need \emph{sub-stochastic relations}: functions $f \colon X\times\Sigma_Y\to[0,1]$ that are measurable in $X$, satisfy $f(x,Y) \leq 1$ for every $x \in X$, and are countably additive in $Y$, \ie $f(x,\bigcup_{i \in \N} U_i) = \sum_{i \in \N} f(U_i)$ for a sequence $U_1,U_2,\ldots$ of disjoint measurable sets.
Thus a stochastic relation (as in Definition~\ref{def:srel}) is a sub-stochastic relation with $f(x,Y)=1$. 
Define a sub-stochastic relation
\[
  \PrEval{-}{-}:\ConC \bA\times \SA{\Rpz \times \ConCV\bA}\to[0,1]
\] 
by $\PrEval CU \defeq \sum_n\PrEvaln n C U$, where 
$\PrEvaln 0 C U$ is given by $[(1,C)\in U]$, and $\PrEvaln {n+1} CU$ is 
\[
  \int_{\Rpz\times\ConC\bA}\hspace{-1cm}\PrEvaln n D {\{(s,E)\,|\,(r\cdot s,E)\in U\}}\ \RedC C {\D(r,D)}.
\]

\begin{proposition}[Termination]\label{prop:fo:termination}
  Evaluation of deterministic terms is a function: $\forall C.\,\exists D.\,C\Downarrow D$.
  Evaluation of probabilistic terms is a stochastic relation: $\forall C.\,\PrEval C {(\Rpz\times \ConCV\bA)}=1$.
\end{proposition}
\begin{proof}
  By induction on the structure of terms.
\end{proof}

Termination comes from the omission of recursion in our language.  We do so for now,
because our semantic model does not yet handle higher-order recursion, and probabilistic 
while-languages are already well-understood~(\eg \cite{kozen:probablistic}). 
(See also the discussion about domain theory in~\S\ref{sec:ho-densem}.) 

\subsection{Soundness}
\label{sec:fo-opsem-sound}
\newcommand{\semvalueconf}{s_{V\icjname}}
For soundness, extend the denotational semantics to configurations:
\begin{itemize}
  \item define $s_{\ipjname} \colon \ConP \bA\to \db \bA$ by $\conf{\Gamma,t}\gamma \mapsto \db t(\gamma)$;
  \item define $s_{\icjname}:\ConC{\bA} \times \Sigma_{\Rpz \times \db \bA} \to [0,1]$ similarly by 
          $(\conf{\Gamma,t}\gamma, U) \,{\mapsto}\, \db t(\gamma)(U)$. We may also use this stochastic relation as a measurable function 
          $s_{\icjname} \colon \ConC{\bA} \,{\to}\, \PM(\Rpz \,{\times}\, \db \bA)$;
  \item define $\semvalueconf:\ConCV{\bA} \to \db\bA$ by ${\conf{\Gamma,\ireturn(v)}\gamma\mapsto \db v(\gamma)}$. Note that 
in this first-order language, $\semvalueconf$ is a surjection
which equates two value configurations iff they are related by weakening, contraction
or exchange of variables.
\end{itemize}

\begin{assumption}
  Throughout this section we assume that the normalization function $\onorm$ on configurations~\eqref{eqn:onorm} is perfect, \ie it corresponds to $\iota$, the semantic normalization function~\eqref{eqn:iota}:
\[\onorm_\bA(\conf {\Gamma,t}\gamma)=\iota_{\db\bA}(s_\icjname(\conf{\Gamma,t}\gamma))\text.\]
\hide{ makes the following diagram commute.
  \[\xymatrix@R-1ex@C-1ex{
        \ConC{\mathbb A}\ar[d]_{s_\icjname}\ar[rr]^-{\onorm_\bA}&&\bigl(\Rpz\times \ConP {\iPM(\mathbb A)}\bigr)+1+1 
        \ar[d]^{(\identity \times s_\ipjname) + \identity + \identity}
        \\ 
        \PM(\Rpz\times \db {\mathbb A})\ar[rr]_-{\iota_{\db\bA}}&&\bigl(\Rpz\times \PM\db{\mathbb A}\bigr)+1+1 
  }\]}
\end{assumption}

\begin{lemma}[Context extension]\mbox{}\label{lem:ctx-ext}
  Let $\ixjname\in\{\ipjname,\icjname\}$. 
  Suppose that $\conf{\Gamma,t}\gamma$ and $\conf {(\Gamma,\Gamma'),t}{(\gamma,\gamma')}$ are configurations in $\ConZ \bA$. 
  Then $s_\ixjname(\conf{\Gamma,t}\gamma)=s_\ixjname\conf{(\Gamma,\Gamma'),t}{(\gamma,\gamma')}$.
\end{lemma}

\begin{proposition}[Soundness]\label{prop:fo:soundness}
  The following diagrams commute (in the category of measurable functions, and stochastic relations, respectively). 
  \[
  \xymatrix@R-6ex@C-2ex{
    \ConPN \bA\ar[dr]^-{s_\ipjname}\ar[dd]_(0.43){\text{one-step}}_(0.57){\text{reduction}}
    \\
    &\db \bA\\
\ConP \bA\ar[ur]_-{s_\ipjname} 
      }
  \xymatrix@R-6ex@C-2ex{
    \ConCN 
    \bA\ar[r]^{s_\icjname}\ar[dd]_(0.43){\text{one-step}}_(0.57){\text{reduction}}
    &
    \Rpz\times \db \bA
    \\\mbox{\phantom{\db\bA}}\\
    \Rpz\times\ConC \bA\ar[r]_-{\identity \times s_\icjname} 
    &
    \Rpz\times\Rpz\times \db\bA\ar[uu]_-{\text{multiplication}}
  }
  \]
\end{proposition}
\begin{proof}
  By induction on the structure of syntax.
  The inductive steps with evaluation contexts use the extension Lemma~\ref{lem:ctx-ext}, which applies by the invariant Lemma~\ref{lemma:type-pres}.
\end{proof}

\paragraph{Adequacy}

The denotational semantics is adequate, in the sense:
\[
        \db{t}(*) = \PM(\Rpz\times\semvalueconf)\big(\PrEval{\conf{\emptyset,t}{*}}{(-)}\big)\quad
        \text{for all $\icj{}{t}{\bA}$}.
\]
That is, the denotation $\db{t}(*)$ is nothing but pushing forward the probability measure $\PrEval{\conf{\emptyset,t}{*}}{(-)}$ from the operational semantics along the function $\semvalueconf$. This adequacy condition holds because 
Proposition~\ref{prop:fo:soundness} ensures that
\[
\bigl(\sum_{k \leq n}\PrEvaln{k}{\conf{\emptyset,t}{*}}{\{(r,C)~|~(r,\semvalueconf(C))\in U\}}\bigr) \leq \db{t}(*)(U)
\]
for all $n$ and $U$, and Proposition~\ref{prop:fo:termination} then guarantees that the left-hand side of this inequality converges to the right-hand side as $n$ tends to infinity.


\section{A higher-order language}
\label{sec:ho-lang}

This short section extends the first-order language with functions and thunks~\cite{levy:thunks}, allowing variables to stand for program fragments. In other words,  `programs are first-class citizens'.
%
%

\paragraph{Types}

Extend the grammar for types with two new constructors.
\begin{align*}
        \bA,\bB & {} \;::=\; 
        \R
        ~|~ \iPM(\bA)
        ~|~ 1
        ~|~ \bA \times \bB
        ~|~ \sum_{i\in I} \bA_i 
        ~|~ \bA\Rightarrow \bB~|~\iMonad(\bA)
\end{align*}
Informally, $\bA\Rightarrow \bB$ contains deterministic functions, and $\iMonad(\bA)$ contains thunked (\ie~suspended) probabilistic programs. Then ${A\Rightarrow \iMonad(\bB)}$ contains probabilistic functions. 
A type is \emph{measurable} if it does not involve $\Rightarrow$ or $\iMonad$, \ie~if it is in the grammar of Section~\ref{sec:fo-lang}. 

\paragraph{Terms}

Extend the term language with the following rules.
First, the usual abstraction and application of deterministic functions:
\begin{align*}
&
\inferback{
\ipj{ \Gamma,x\colon \bA} t \bB
}{
\ipj \Gamma {\lambda x.\,t} {\bA\Rightarrow \bB}
}
\qquad\qquad
\inferback{\ipj \Gamma t {\bA\Rightarrow \bB}
\quad
\ipj \Gamma u \bA
}{
\ipj \Gamma {t\, u} {\bB}
}
\\
\intertext{Second, we have syntax for thunking and forcing (\eg \cite{levy:thunks,moggi-monads,parkpfenningthrun:sampling}).}
&\inferback{
\icj \Gamma t \bA
}{
\ipj \Gamma {\ithunk(t)} {\iMonad(\bA)}
}
\qquad\qquad
\inferback{
\ipj \Gamma t {\iMonad(\bA)}
}{
\icj \Gamma {\iforce(t)} {\bA}
}
\end{align*}
All the rules from Section~\ref{sec:fo-lang} are also still in force. For rule~\eqref{eq:functionterms} to still make sense, we only include constant terms for measurable functions $f \colon \db\bA \to \db\bB$ between measurable types $\bA$ and $\bB$.

\paragraph{Examples}
One motivation for higher types is to support code
structuring. The separate function types and thunk types allow
us to be flexible about calling conventions.
For example, sampling can be reified as the ground term
\[
  \ipj{}{\lambda x.\,\ithunk(\isample(x))}{\iPM(\bA)\Rightarrow \iMonad(\bA)},
\]
which takes a probability measure and returns a suspended program that will sample from it.
On the other hand, to reify the normalization construction,
we use a different calling convention.
\[
  \ipj{}{\lambda x.\,\inorm(\iforce(x))}{\iMonad(\bA)\Rightarrow \R\times \iPM(\bA)+1+1}
\]
This function takes a suspended probabilistic program and returns the result of normalizing it.

\paragraph{Example: higher-order expectation}
Higher types also allow us to consider probability distributions
over programs. 
For an example, consider this term.
\[
E_h \;\;\defeq\;\;
\begin{array}[t]{@{}l@{}}
\lambda (d,f) \colon \iMonad(\bA) \times (\bA\Rightarrow \R).\,
\\
\qquad
\begin{array}[t]{@{}r@{}l@{}}
\icase\, & (\inorm(\ilet\, a = \iforce(d)\, \iin\, \iscore(f(a))))\ \iof
\\
   & (0,(e,y)) \Rightarrow e
\\
|\ & (1,*) \Rightarrow 0.0
\ \ |\  (2,*) \Rightarrow 0.0
\end{array}
\end{array}
\]
It has type $(\iMonad(\bA) \times (\bA \Rightarrow \R)) \Rightarrow \R$.
Intuitively, given a thunked probabilistic term $t$ and a function $f$ that is nonnegative, 
$E_h$ treats $t$ as a probability distribution on $\bA$, and computes the expectation 
of $f$ on this distribution. Notice that $\bA$ can be a higher type, so $E_h$ 
generalises the usual notion of expectation, which has not been defined for higher
types because the category of measurable spaces is not Cartesian closed.


\section{Higher-order operational semantics}
\label{sec:ho-opsem}

\newcommand*{\ixctx}[2]{#1\mathrel{\vdash\!\!\!\!_\mathsf{z}}#2}
\newcommand*{\ipctx}[2]{#1\mathrel{\vdash\!\!\!\!_\mathsf{d}}#2}
\newcommand*{\icctx}[2]{#1\mathrel{\vdash\!\!\!\!_\mathsf{p}}#2}
\newcommand{\LRedZ}[2]{\mathrm{R}(\ixctx{#1}{#2})}
\newcommand{\LRedZV}[2]{\mathrm{R}_\mathrm{v}(\ixctx{#1}{#2})}
\newcommand{\LRedP}[2]{\mathrm{R}(\ipctx{#1}{#2})}
\newcommand{\LRedPV}[2]{\mathrm{R}_\mathrm{v}(\ipctx{#1}{#2})}
\newcommand{\LRedC}[2]{\mathrm{R}(\icctx{#1}{#2})}
\newcommand{\LRedCV}[2]{\mathrm{R}_\mathrm{v}(\icctx{#1}{#2})}

In this section we consider operational semantics for the higher-order extension of the language.
In an operational intuition, $\iforce(t)$ forces a suspended computation $t$ to run.
For example, 
\[
  \ipj{}{\ithunk(\isample(\norDS(0.0,1.0)))}{\iMonad(\R)}
\] 
is a suspended computation that, when forced, will sample the normal distribution.

\hide{
It is often convenient for an anonymous function to suspend its body, 
so that applying it to an argument causes the body to evaluate.
This convention is encoded by the following derived constructs.
\[
  \inferback{
    \icj{ \Gamma,x\colon \bA} t \bB
  }{
    \ipj \Gamma {\lambda x.\,\ithunk(t)} {\bA\Rightarrow \iMonad(\bB)}
  }
  \quad
  \inferback{\ipj \Gamma t {\bA\Rightarrow \iMonad(\bB)}
  \quad
  \ipj \Gamma u \bA
  }{
    \icj \Gamma {\iforce(t\, u)} {\bB}
  }
\]
}

\begin{assumption}
\label{ass:no-P-ho}
From the operational perspective it is unclear how to deal with
sampling from a distribution over functions. For this reason, in this section, we
only allow the type $\iPM(\bA)$ when $\bA$ is a measurable type.
We still allow probabilistic terms to have higher-order types, 
and we still allow $\iMonad(\bA)$ where $\bA$ is higher-order.
\end{assumption}
\subsection{Reduction}
\label{sec:ho-opsem-red}

We now extend the operational semantics from Section~\ref{sec:fo-opsem} with higher types.
Values~\eqref{eqn:fo-value} are extended as follows.
\[ 
        v \,::=\,\dots~|~  \lambda x.t~|~\ithunk(t) 
\]
Evaluation contexts~\eqref{eqn:fo-ctx} are extended as follows.
\[
        \ctx{-} \,::=\, \dots      ~|~ \ctx{-}\,t~|~ v\,\ctx{-}
        \qquad
        \ectxp{-} \,::=\, \dots~ |~\ectxc{\iforce[-]}
\]
There are two additional redexes:
$(\lambda x.t)\,v$ and 
$\iforce(\ithunk(t))$.
The deterministic transition relation is extended with this $\beta$-rule:
\[
        \conf {\Gamma,(\lambda x.t)\,v}\gamma \longrightarrow \conf{\Gamma,t[v/x]}{\gamma}.
\]
Extend the probabilistic transition relation with the following rules.
\begin{align*}
&
\bigRedCI{\conf {\Gamma,\iforce(\ithunk(t))}\gamma}{\Gamma,t}=1
\\
&\bigRedCC{\conf {\Gamma,\iforce(\ithunk(t))}\gamma}{(\Gamma,t)}U=\indicator{\gamma\in U}
\\
&\Score({\conf {\Gamma,\iforce(\ithunk(t))}\gamma})=1
\end{align*}

\subsection{Termination}\label{subsec:ho-termination}

The evaluation relations for deterministic and probabilistic configurations of the higher-order language are defined as in Subsection~\ref{subsec:fo-termination}. 
The resulting rewriting system still terminates, even though configurations may now include higher-order terms.

\begin{proposition}[Termination]\label{prop:ho:termination} 
        Evaluation of deterministic terms is a function: $\forall C.\,\exists D.\,C\Downarrow D$. Evaluation of probabilistic terms is a stochastic relation: $\forall C.\,\PrEval C {(\Rpz\times \ConCV \bA)}=1$. 
\end{proposition}
\begin{proof}
  We sketch an invariant of higher-order terms that implies the termination property, formulated as unary logical relations via sets 
  \begin{align*}
    \LRedZ \Gamma \bA&\subseteq\{t~|~\ixj\Gamma t \bA\}, \\
    \LRedZV \Gamma \bA&\subseteq\{t~|~\ixj\Gamma t \bA \wedge
    t\text{ a $\ixjname$-value}\},
  \end{align*}
  for each canonical context $\Gamma$, type $\bA$, and $\ixjname \in \{\ipjname,\icjname\}$, defined by:
  \begin{align*}
    \LRedP \Gamma \bA
    &= \{t \mid \forall \gamma.\,\conf {\Gamma,t}\gamma\Downarrow \conf{\Gamma',t'}{\gamma'}
      \wedge t'\in\LRedPV {\Gamma'} \bA\} \\
    \LRedC \Gamma \bA
    &= \{t \mid \forall \gamma.\,\Prob{\conf{\Gamma,t} \gamma \Downarrow 
      \\& \qquad\qquad\quad(\Rpz\times \textstyle{\sum\nolimits_{\Gamma'}}\LRedCV{\Gamma'} \bA\times \db{\Gamma'}})=1\} 
      \end{align*}\begin{align*}
    \LRedCV \Gamma \bA
    &=\{\ireturn(v) \mid v\in\LRedPV \Gamma \bA\} \\
    \LRedPV \Gamma {\bA}
    &=\{x \mid (x \colon \bA)\in\Gamma\} \qquad\text{for $\bA$ indecomposable} \\
    \LRedPV \Gamma{1} &=\{*\}\\
    \LRedPV \Gamma{\bA_1 \times \bA_2}
    &=\{(v_1,v_2) \mid \forall j.\,v_j\in\LRedPV \Gamma{\bA_j}\} 
    \\
    \LRedPV \Gamma{\sum \bA_i}
    &=\{(i,v) \mid v\in\LRedPV \Gamma{\bA_i}\} \\
    \LRedPV \Gamma{T(\bA)}
    &= \{{\ithunk(t)} \mid t\in\LRedC \Gamma \bA\}\\
    \LRedPV \Gamma{\bA\Rightarrow \bB}
    &= \{{\lambda x.t} \mid \forall \Gamma'\supseteq \Gamma, u\in\LRedPV {\Gamma'} \bA.\\
    & \phantom{=\{{\lambda x.t} \mid \forall}{t[u/x]}\in\LRedP {\Gamma'} \bB\}
  \end{align*}
  Induction on the structure of a term $\ixj{\Gamma,x_1\colon \bA_1,\dots x_n\colon \bA_n}t \bB$ for $\ixjname \in \{\ipjname,\icjname\}$ now proves that $v_i \in \LRedPV \Gamma{\bA_i}$ for $i=1,\ldots,n$ implies $t[\vec v/\vec x]\in\LRedZ \Gamma \bB$.
\end{proof}



\section{Higher-order denotational semantics}
\label{sec:ho-densem}

\newcommand{\cat}[1][C]{\mathbf{#1}}
\newcommand{\yoneda}{\mathbf{y}}
\newcommand{\kan}[1]{\overline{#1}}

This section gives denotational semantics for the higher-order language, without using Assumption~\ref{ass:no-P-ho}.
We are to interpret the new constructs $\iMonad(\bA)$, $\ithunk$, and $\iforce$.
We will interpret probabilistic judgements as Kleisli morphisms $\db\Gamma\to \Monad(\db A)$ for a certain monad~$\Monad$, and set $\db{\iMonad(A)}\defeq \Monad(\db A)$, so $\ithunk$ and $\iforce$ embody the correspondence of maps $\db\Gamma\to \Monad(\db A)$ and $\db\Gamma\to \db {\iMonad(A)}$. 

On which category can the monad $\Monad$ live?
Interpreting $\lambda$-abstraction and application needs a natural `currying' bijection between morphisms $\db\Gamma\times \R\to\R$ and morphisms $\db\Gamma\to \db{\R\Rightarrow\R}$.
But measurable functions cannot do this: it is known that no measurable space $\db{\R\Rightarrow\R}$ can support such a bijection~\cite{aumann:functionspaces}.

We resolve the problem of function spaces by embedding the category of measurable spaces in a larger one, where currying is possible, and that still has the structure to interpret the first order language as before.
As the larger category we will take a category of functors $\Meas\op\to\Set$ from a category $\Meas$ of measurable spaces and measurable functions to the category $\Set$ of sets and functions.
This idea arises from two traditions.
First, we can think of a variable of type $\R$ as a read-only memory cell, as in the operational semantics, and functor categories have long been used to model local memory~(\eg \cite{oles}).
Second, the standard construction for building a Cartesian closed category out of a distributive one is based on functor categories~(\eg \cite{power-genericmodels}).

\paragraph{Other models of higher-order programs}

Semantics of higher-order languages with \emph{discrete} probability are understood well. 
For terminating programs, there are set-theoretic models based on a distributions monad, and for full recursion one can use probabilistic powerdomains~\cite{jp-prob-powerdomain} or coherence spaces~\cite{etp-pcoh}.
It is also plausible one could model continuous distributions in domain theory, since it supports computable real analysis~(\eg~\cite{ee-integration}); this could be interesting because computability is subtle for probabilistic programming~(\eg \cite{afr-noncomputable}).
Nonetheless, we contend it is often helpful to abstract away computability issues when studying probabilistic programming languages, to have access to standard theorems of probability theory to justify program transformations.

\subsection{Semantic model}

Fix a category $\Meas$ of measurable spaces and measurable functions that is essentially small but large enough for the purposes of Section~\ref{sec:fo-densem}. For example, $\Meas$ could be the category of \emph{standard Borel spaces}~\cite{berberian:borel,srivastava:borel}: one can show that $\db \bA$ is standard Borel by induction on $\bA$, 
and the class of all standard Borel spaces is countable up to measurable isomorphism.

In Section~\ref{sec:fo-densem} we interpreted first-order types $\bA$ as measurable spaces $\db \bA$. 
We will interpret higher-order types~$\bA$ as functors $\dho\bA\colon \Meas\op\to \Set$.
Informally, when $\bA$ is a first-order type and $\Gamma$ is a first-order context, we will have
$\dho\bA(\db\Gamma)\cong \Meas(\db\Gamma,\db\bA) \approx\{t~|~\ipj\Gamma t \bA\}$.
For a second order type $(\bA\Rightarrow \bB)$, we will have
\[\dho{\bA\Rightarrow\bB}(\db\Gamma)\cong \Meas(\db\Gamma\times \db\bA,\db\bB) \approx\{t~|~\ipj{\Gamma,x:\bA} t \bB\}\]
so that $\beta/\eta$ equality is built in.
To put it another way, $\dho\bA(\R^n)$~models terms of type $\bA$ having $n$~read-only real-valued memory cells.


\begin{lemma}\label{lemma:orthogonality}
  For a small category $\cat$ with countable sums, consider the category $\kan{\cat}$ of countable-product-preserving functors $\cat\op\to\Set$, and natural transformations between them.
  \begin{itemize}
  \item $\kan\cat$ has all colimits;
  \item $\kan\cat$ is Cartesian closed if $\cat$ has products that distribute over sums;
  \item There is a full and faithful embedding $\yoneda:\cat\to\kan \cat$,
    given by $\yoneda(c)\defeq \cat(-,c)$,
    which preserves limits and countable sums.
  \end{itemize}
\end{lemma}
\begin{proof}
See \eg\cite[\S7]{power-genericmodels}, or \cite[Theorems~5.56 and~6.25]{kelly-book}.
  The embedding $\yoneda$ is called the \emph{Yoneda embedding}.
\end{proof}

For a simple example, consider the category $\cat[CSet]$ of countable sets and functions. It has countable sums and finite products, but is not Cartesian closed. Because every countable set is a countable sum of singletons, the category $\kan{\cat[CSet]}$ is equivalent to $\Set$.

Our semantics for the higher-order language will take place in the category $\kan{\Meas}$. 
Note that products in $\kan \Meas$ are pointwise, \eg $(F\times G)(X)=F(X)\times G(X)$ for all $F,G\in\kan\Meas$ and all $X\in\Meas$, but sums are not pointwise, \eg $(1+ 1)\in\kan\Meas$ is the functor that assigns a measurable space $X$ to the set of its measurable subsets. This is essential for $\yoneda$ to preserve sums.

\paragraph{Distribution types}

We have to interpret distribution types $\iPM(\bA)$ in our functor category $\kan\Meas$. 
How can we interpret a probability distribution on the type $\R\Rightarrow\R$?
We can answer this pragmatically, without putting $\sigma$-algebra structure on the set of all functions. 
If $\db{\R\Rightarrow\R}$ were a measurable space, a random variable valued in $\db{\R\Rightarrow\R}$ would be given by a measurable space $(X,\Sigma_X)$, a probability distribution on it, and a measurable function $X \to \db{\R\Rightarrow\R}$. 
Despite there being no such measurable space $\db{\R\Rightarrow \R}$, we can speak of uncurried measurable functions $X \times \R\to \R$.
Thus we might define a probability distribution on $\db{\R\Rightarrow \R}$ to be a triple
\[
  \big((X,\Sigma_X),\, f\colon X\times \R\to \R,\, p\colon \Sigma_X\to[0,1]\big)
\]
of a measurable space $(X,\Sigma_X)$ of `codes', a measurable function $f$ where we think of $f(x,r)$ as `the function coded $x$ evaluated at $r$',  and a probability distribution $p$ on the codes. 
These triples should be considered modulo renaming the codes.
This is exactly the notion of probability distribution that arises in our functor category.

\begin{lemma}\label{lemma:kan}
  For a small category $\cat$ with countable sums:
  \begin{itemize}
    \item any functor $F\colon \cat\to\cat$ extends to a functor $\kan{F} \colon \kan{\cat} \to \kan{\cat}$ satisfying $\kan{F}\circ \yoneda \cong \yoneda \circ F$, given by
      \[  
        \kan{F}(G)(b) =
        \Big(\sum\nolimits_aG(a)\times \cat\big(b,F(a)\big)\Big) \;\slash\,\!\mathop{\sim}
      \]
	  where the equivalence relation $\sim$ is the least one satisfying ${(a',x,Fg\circ f)}\sim {(a,Gg(x),f)}$;
     \item similarly, any functor $F\colon \cat\times\cat\to\cat$ in two arguments
       extends to  a functor $\kan F\colon \kan\cat\times \kan\cat 
       \to \kan \cat$, with $\kan F\circ (\yoneda\times \yoneda)\cong\yoneda\circ F$:
       \[
       \kan{F}(G,H)(c)=
        \Big(\textstyle{\sum\nolimits_{a,b}G(a)\times H(b)\times\cat\big(c,F(a,b)\big)}\Big) \;\slash\,\!\mathop{\sim}
       \]
    \item any natural transformation $\alpha \colon F\to G$ between functors $F,G\colon\cat\to\cat$ lifts to a natural transformation $\kan{\alpha} \colon \kan{F} \to \kan{G}$, and similarly for functors $\cat \times \cat \to \cat$;
    \item and this is functorial, \ie
      $\kan{G \circ F} \cong \kan{G} \circ \kan{F}$ and 
      $\kan{\beta \circ \alpha} = \kan{\beta} \circ \kan{\alpha}$.
  \end{itemize}	
\end{lemma}
\begin{proof}
  $\kan F (G)$ is the left \emph{Kan extension} of $G$ along $F$, see \eg~\cite{kelly-book}.
  Direct calculation shows $\kan F(G)$ preserves products if $G$ does.
\end{proof}

Thus the commutative monads $\PM$ and $\Monad=\PM(\Rpz\times (-))$ on $\Meas$ lift to commutative monads $\kan \PM$ and $\kan \Monad \cong \kan\PM(\yoneda\Rpz\times (-))$ on $\kan\Meas$. 
The latter monad captures the informal importance-sampling semantics advocated by the designers of Anglican~\cite{wood-aistats-2014}.

\subsection{Conservativity}

We interpret the types of the higher order language as objects in $\kan \Meas$ using its categorical structure. 
\begin{align*}
  \textstyle{\dho {\sum_{i\in I} \bA_i}} & \defeq \textstyle{\sum_{i\in I}\dho{\bA_i} } &
  \dho \R & \defeq\yoneda\R &
  \dho {\iPM(\bA)} & \defeq\kan\PM(\dho \bA) \\
  \dho {\bA\times \bB} & \defeq \dho\bA\times \dho\bB &
  \dho {1} & \defeq 1 &
  \dho {\iMonad(\bA)} & \defeq\kan\Monad(\dho \bA) \\
  \dho{\bA\Rightarrow \bB} & \defeq \dho{\bA}\Rightarrow \dho{\bB} 
\end{align*}
We extend this interpretation to contexts.
\[
  \dho{x_1 \colon \bA_1,\ldots,x_n \colon \bA_n}
  \defeq
  \textstyle{\prod_{i=1}^n\dho{\bA_i}}
\]
Deterministic terms $\ipj{\Gamma}{t}{\bA}$ are interpreted as natural transformations 
$\dho\Gamma\to\dho \bA$ in $\kan\Meas$, and probabilistic terms ${\icj{\Gamma}{t}{\bA}}$ as natural transformations $\dho\Gamma\to\kan T\dho\bA$, by induction on the structure of terms as in Section~\ref{sec:fo-densem}. 
Application and $\lambda$-abstraction are interpreted as usual in Cartesian closed categories~\cite{pitts-catlogic}. 
Thunk and force are trivial from the perspective of the denotational semantics, because $\dho{\iMonad \bA}=\Monad\dho\bA$. 
To interpret $\inorm(t)$, use Lemma~\ref{lemma:kan} to extend the normalization functions $\Monad(X)\to \Rpz\times \PM(X)+1+1$ between measurable spaces~\eqref{eqn:iota} to natural transformations $\kan \Monad(F)\to \yoneda(\Rpz)\times \kan P(F)+1+1$. 

Note that all measurable types $\bA$ have a natural isomorphism $\dho\bA \cong \yoneda\db\bA$. 
This interpretation conserves the first-order semantics of Section~\ref{sec:fo-densem}:

\begin{proposition}
\label{prop:church-turing}
  For $\ixjname \in \{\ipjname,\icjname\}$, and first-order $\Gamma$ and $\bA$:
  \begin{itemize}
    \item  for first-order $\ixj \Gamma {t,u} \bA$, \ $\db t = \db u$ if and only if $\dho t = \dho u$;
    \item every term $\ixj \Gamma t \bA$ has $\dho t = \dho u$ for a first-order $\ixj \Gamma u \bA$.
  \end{itemize}
\end{proposition}
\begin{proof}
  We treat $\ixjname=\ipjname$; the other case is similar.
  By induction on the structure of terms, $\dho t=\yoneda\db t \colon \dho \Gamma\to\dho \bA$.
  The first point follows from faithfulness of $\yoneda$;
  the second from fullness and~\eqref{eq:functionterms}.
\end{proof}

One interesting corollary is that the interpretation of a term of first-order type is always a measurable function, even if the term involves thunking, $\lambda$-abstraction and application. This corollary gives a partial answer to a question by Park \textit{et.\ al.}\ on the measurability of all $\lambda$-definable ground-type terms in probabilistic programs~\cite{parkpfenningthrun:sampling} (partial because our language does not include recursion).

\subsection{Soundness}
\label{sec:ho-sound}

The same recipe as in Section~\ref{sec:fo-opsem-sound} will show that the higher-order denotational semantics is sound and adequate with respect to the higher-order operational semantics. This needs Assumption~\ref{ass:no-P-ho}.

A subtle point is that configuration spaces~\eqref{eq:configurationspace} involve uncountable sums: the set of terms of a given type is uncountable, but $\yoneda$ only preserves countable sums. 
This is not really a problem because only countably many terms are reachable from a given program.
\newcommand{\reachp}{\rightsquigarrow^*_\ipjname}
\newcommand{\reachc}{\rightsquigarrow^*_\icjname}
\newcommand{\reachx}{\rightsquigarrow^*_\ixjname}
\newcommand{\TermUniverse}{\mathcal U}

\begin{definition}
  For a type $\bA$, the binary \emph{reachability} relation $\reachp$ on $\{(\Gamma,t) \mid \ipj\Gamma t \bA\wedge\Gamma\,\text{canonical}\}$ is the least reflexive and transitive relation with $(\Gamma,t)\reachp(\Gamma',u)$ if $\conf{\Gamma,t}\gamma\longrightarrow \conf{\Gamma',u}{\gamma'}$ for $\gamma\in\db\Gamma$, $\gamma'\in\db{\Gamma'}$.
  Similarly, $\reachc$ is the least reflexive and transitive relation on ${\{(\Gamma,t) \mid \icj\Gamma t \bA\wedge\Gamma\,\text{canonical}\}}$ with $(\Gamma,t)\reachc(\Gamma',u)$ if $\bigRedCI{\conf{\Gamma,t}\gamma}{\Gamma',u}\neq 0$ for $\gamma\in\db\Gamma$.
\end{definition}

\begin{proposition}
  Let $\ixjname \in \{\ipjname, \icjname\}$. 
  For any closed term $\ixj{}t\bA$, the set of reachable terms $\{(\Gamma,u)~|~(\emptyset,t)\reachx (\Gamma,u)\}$ is countable.
\end{proposition}
\begin{proof}
  One-step reachability is countable by induction on terms.
  Since all programs terminate by Proposition~\ref{prop:ho:termination},
  the reachable terms form a countably branching well-founded tree.
\end{proof}

We may thus restrict to the configurations built from a countable set~$\TermUniverse$ of terms that is closed under subterms and reachability.
Extend the denotational semantics in $\kan\Meas$ to configurations by defining
${s_\ipjname \colon \yoneda(\ConP\bA)\to\dho\bA}$,
${s_\icjname \colon \yoneda(\ConC\bA)\to \kan\Monad\dho\bA}$, and 
${\semvalueconf \colon \yoneda(\ConCV\bA)\to \kan\Monad\dho\bA}$;
use the isomorphisms
\[\yoneda(\ConZ\bA)\cong
        \sum_{
                (\Gamma \vdash_{\!\!\mathsf{z}}\;t : \bA)\in\TermUniverse}
        \dho{\Gamma}
 \]
to define $s_\ipjname,s_\icjname,\semvalueconf$ by copairing the interpretation morphisms
$\dho {\ipj\Gamma t\bA} \colon \dho\Gamma\to\dho \bA$ and ${\dho {\icj\Gamma u\bA} \colon \dho\Gamma\to \kan T(\dho \bA)}$.

\begin{proposition}[Soundness]\label{prop:ho:soundness}
  The following diagrams commute.
  \[
  \xymatrix@R-6ex@C-2ex{
    \yoneda(\ConPN \bA)\ar[dr]^-{s_\ipjname}\ar[dd]_{\yoneda(\text{reduction})}\\
    &\dho \bA
    \\
    \yoneda(\ConP \bA)\ar[ur]_-{s_\ipjname} \\
  }\ \ 
    \xymatrix@R-6ex@C-2ex{
    \yoneda(\ConCN \bA)\ar[dr]^-{s_\icjname}\ar[dd]_-{\yoneda(\text{reduction})}
    \\
    &\kan T\dho \bA
    \\
    \kan T(\ConC \bA)\ar[ur]_-{\kan\mu\cdot\kan T(s_{\icjname})}
  }
  \]
\end{proposition}

\paragraph{Adequacy} 

It follows that the higher denotational semantics remains adequate, in the sense that for all probabilistic terms $\icj {}t \bA$, 
\[
  \dho{t}_1(*) = (\kan T(\semvalueconf))_1\big(\PrEval{\conf{\emptyset,t}{*}}{(-)}\big)\text.
\]
Adequacy is usually only stated at first-order types. 
At first-order types~$\bA$ the function $\semvalueconf$ does very little, 
since global elements of $\dho \bA$ correspond bijectively with value configurations modulo weakening, contraction and exchange in the context.
At higher types, the corollary still holds, but $\semvalueconf$ is not
so trivial because we do not reduce under $\ithunk$ or $\lambda$. (See
also~\cite{pp-adequacy}.)


\section{Continuous densities}
\label{sec:density}

In most examples, the argument $t$ to $\iscore(t)$ is a density function
for a probability distribution. When scores are based on density
functions, this makes the relationship
between $\iscore$ and $\isample$ tighter than we have expressed so far. 
Our language easily extends to accommodate such distributions. We just add a collection \emph{density types} to the syntax.
\[
      \bD \,::=\, \R \mid \iboolty \mid \N \mid 1 \mid \bD \times \bD
      \qquad\quad
      \bA \,::=\, \cdots \mid \density(\bD)
\]
The $\bD$ in this grammar denotes a measurable space that:
(i) carries a separable metrisable topology that generates the $\sigma$-algebra;
and that (ii) comes with a chosen $\sigma$-finite measure $\mu_\bD$.
An example is $\R$ with its usual Euclidean topology and the Lebesgue
measure (which maps each interval to its size).

The type $\density(\bD)$ denotes a measurable space $\db{\density(\bD)}$ of continuous functions $f \colon \db{\bD} \to \Rpz$ with $\int_X f\, \D \mu_\bD = 1$.
The $\sigma$-algebra of $\db{\density(\bD)}$ is the least one making $\{f \mid f(x) \leq r\}$ measurable for all $(x,r) \in \db{\bD} \times \Rpz$.

A density type $\density(\bD)$ comes with two measurable functions
\begin{align*}
  &\mathit{ev} \colon \db{\density(\bD)} \times \db{\bD} \to \Rpz
        &&
        \mathit{ev}(f,x) = f(x)\\
    &\mathit{dist}\colon \db{\density(\bD)}\to \db{\iPM(\bD)}
      &&\mathit{dist}(f)(U)=\int_Uf\,\D \mu_{\bD}
\end{align*}
Note that measurability of $\mathit{ev}$ relies on continuity of the densities~\cite{aumann:functionspaces}. 
The usual way of
imposing a soft constraint based on a likelihood can be encoded in our first-order language as
$\iscore(\mathit{ev}(f,x))$,
where the datum $x$ is observed with a probability distribution with density $f$. 
Thus the categorical machinery used to interpret higher-order functions is not needed for such soft constraints. 

There is a limit to this use of continuity: 
probability measures produced by $\inorm(t)$ need not have continuous density.
For example, $\inorm(\ireturn(42.0))$ produces a discontinuous Dirac measure.

Probability densities are often used by \emph{importance samplers}. 
The importance sampler generates samples of $f \in \db{\density(\R)}$ 
by first sampling from a proposal distribution $g$ where sampling is easy, and then
normalizing those samples $x$ from $g$ according to their
importance weight $f(x)/g(x)$. This works provided the support of $g$
is $\R$, e.g.~$g=\mathit{density}{\_}\norDS(x,(0.0,1.0))$.
\begin{align*} 
        &
        \db{\inorm(\isample(\mathit{dist}(f)))} =
        \\
        & 
        \left\llbracket
\inorm(\ilet\,x=\isample(\mathit{dist}(g))\,\iin\,\iscore(\textstyle{\frac{\mathit{ev}(f,x)}{\mathit{ev}(g,x)}});\ireturn(x))
        \right\rrbracket
\end{align*}

Density types can be incorporated into the higher order language straightforwardly. 
The only subtlety is that denotational semantics now needs the base category to contain $\db{\density(\bD)}$.

\section{Conclusion and future work}
\label{sec:future}

We have defined a metalanguage for higher-order probabilistic programs with continuous distributions and soft constraints, 
and presented operational and denotational semantics, 
together with useful program equations justified by the semantics. 
One interesting next step is to use these tools to study other old or new language features and concepts 
(such as recursion, function memoisation~\cite{Roy-pp-nonparametric2008},
measure-zero conditioning~\cite{borgstrometal:bayesiantransformer}, 
disintegration~\cite{Shan-disintegration2016,Ackerman-disintegration2015},
and exchangeability~\cite{FreerR12,Mansinghka-venture14,wood-aistats-2014}) 
that have been experimented with in the context of probabilistic programming. 
Another future direction is to formulate and prove
the correctness of inference algorithms, especially those based on
Monte Carlo simulation, following~\cite{HurNRS15}.

\acks

We thank T.\ Avery, I.\ Garnier, T.\ Le, K.\ Sturtz, and A.\ Westerbaan.
This work was suppoted by the EPSRC, a Royal Society University Fellowship,
An Institute for Information \& communications Technology Promotion (IITP) 
grant funded by the Korea government (MSIP, No. R0190-15-2011), 
DARPA PPAML, and
the ERC grant `causality and symmetry --- the next-generation semantics'.

\bibliographystyle{abbrvnat}
\bibliography{lics2016}

\begin{thebibliography}{35}
\providecommand{\natexlab}[1]{#1}
\providecommand{\url}[1]{\texttt{#1}}
\expandafter\ifx\csname urlstyle\endcsname\relax
  \providecommand{\doi}[1]{doi: #1}\else
  \providecommand{\doi}{doi: \begingroup \urlstyle{rm}\Url}\fi

\bibitem[Ackerman et~al.(2011)Ackerman, Freer, and Roy]{afr-noncomputable}
N.~L. Ackerman, C.~E. Freer, and D.~M. Roy.
\newblock Noncomputable conditional distributions.
\newblock In \emph{LiCS}, 2011.

\bibitem[Ackerman et~al.(2015)Ackerman, Freer, and
  Roy]{Ackerman-disintegration2015}
N.~L. Ackerman, C.~E. Freer, and D.~M. Roy.
\newblock On computability and disintegration, 2015.
\newblock URL \url{http://arxiv.org/abs/1509.02992}.

\bibitem[Aumann(1961)]{aumann:functionspaces}
R.~J. Aumann.
\newblock Borel structures for function spaces.
\newblock \emph{Illinois Journal of Mathematics}, 5:\penalty0 614--630, 1961.

\bibitem[Berberian(1988)]{berberian:borel}
S.~K. Berberian.
\newblock \emph{Borel spaces}.
\newblock World Scientific, 1988.

\bibitem[Borgstr{\"{o}}m et~al.(2013)Borgstr{\"{o}}m, Gordon, Greenberg,
  Margetson, and {van Gael}]{borgstrometal:bayesiantransformer}
J.~Borgstr{\"{o}}m, A.~D. Gordon, M.~Greenberg, J.~Margetson, and J.~{van
  Gael}.
\newblock Measure transformer semantics for {B}ayesian machine learning.
\newblock \emph{LMCS}, 9\penalty0 (3):\penalty0 11, 2013.

\bibitem[Borgstr{\"{o}}m et~al.(2015)]{borgstrometal:op-sem}
J.~Borgstr{\"{o}}m, U.~Dal~Lago, A.~D.~Gordon, M.~Szymczak.
\newblock A lambda-calculus foundation for universal probabilistic programming.
\newblock arxiv:1512.08990. 2015.

\bibitem[Doberkat(2007)]{doberkat-book}
E.-E. Doberkat.
\newblock \emph{Stochastic Relations: Foundations for Markov Transition
  Systems}.
\newblock Chapman \& Hall, 2007.

\bibitem[Doucet et~al.(2001)Doucet, {de Freitas}, and
  Gordon]{doucetfreitasgordon:montecarlo}
A.~Doucet, N.~{de Freitas}, and N.~Gordon.
\newblock \emph{Sequential Monte Carlo Methods in Practice}.
\newblock Springer, 2001.

\bibitem[Edalat and {Escard\'o}(2000)]{ee-integration}
A.~Edalat and M.~H. {Escard\'o}.
\newblock Integration in {R}eal {P}{C}{F}.
\newblock \emph{Inf.~Comput.~}, 160:\penalty0 128--166, 2000.

\bibitem[Ehrhard et~al.(2014)Ehrhard, Tasson, and Pagani]{etp-pcoh}
T.~Ehrhard, C.~Tasson, and M.~Pagani.
\newblock Probabilistic coherence spaces are fully abstract for probabilistic
  {P}{C}{F}.
\newblock In \emph{POPL}, 2014.

\bibitem[Freer and Roy(2012)]{FreerR12}
C.~E. Freer and D.~M. Roy.
\newblock Computable de {F}inetti measures.
\newblock \emph{Ann. Pure Appl. Logic}, 163\penalty0 (5):\penalty0 530--546,
  2012.

\bibitem[Giry(1982)]{giry:monad}
M.~Giry.
\newblock A categorical approach to probability theory.
\newblock \emph{Categorical Aspects of Topology and Analysis}, 915:\penalty0
  68--85, 1982.

\bibitem[Goodman et~al.(2008)Goodman, Mansinghka, Roy, Bonawitz, and
  Tenenbaum]{goodman_uai_2008}
N.~Goodman, V.~Mansinghka, D.~M. Roy, K.~Bonawitz, and J.~B. Tenenbaum.
\newblock {Church: a language for generative models}.
\newblock In \emph{UAI}, 2008.

\bibitem[Hur et~al.(2015)Hur, Nori, Rajamani, and Samuel]{HurNRS15}
C.~Hur, A.~V. Nori, S.~K. Rajamani, and S.~Samuel.
\newblock A provably correct sampler for probabilistic programs.
\newblock In \emph{FSTTCS}, 2015.

\bibitem[Jones and Plotkin(1989)]{jp-prob-powerdomain}
C.~Jones and G.~D. Plotkin.
\newblock A probabilistic powerdomain of evaluations.
\newblock In \emph{LiCS}, 1989.

\bibitem[Kelly(1982)]{kelly-book}
G.~M. Kelly.
\newblock \emph{Basic concepts of enriched category theory}.
\newblock CUP, 1982.

\bibitem[Kock(1970)]{kock:commutative}
A.~Kock.
\newblock Monads on symmetric monoidal closed categories.
\newblock \emph{Archiv der Mathematik}, XXI:\penalty0 1--10, 1970.

\bibitem[Kozen(1981)]{kozen:probablistic}
D.~Kozen.
\newblock Semantics of probablistic programs.
\newblock \emph{Journal of Computer and System Sciences}, 22:\penalty0
  328--350, 1981.

\bibitem[Levy()]{levy:thunks}
P.~B. Levy.
\newblock Call-by-push-value: A subsuming paradigm.
\newblock In \emph{TLCA'02}.

\bibitem[Levy et~al.(2003)Levy, Power, and Thielecke]{lpt-cbv}
P.~B. Levy, J.~Power, and H.~Thielecke.
\newblock Modelling environments in call-by-value programming languages.
\newblock \emph{Inf.~Comput.}, 185\penalty0 (2), 2003.

\bibitem[Mansinghka et~al.(2014)Mansinghka, Selsam, and
  Perov]{Mansinghka-venture14}
V.~K. Mansinghka, D.~Selsam, and Y.~N. Perov.
\newblock Venture: a higher-order probabilistic programming platform with
  programmable inference.
\newblock 2014.
\newblock URL \url{http://arxiv.org/abs/1404.0099}.

\bibitem[Minka et~al.(2010)Minka, Winn, Guiver, and
  Knowles]{minka_software_2010}
T.~Minka, J.~Winn, J.~Guiver, and D.~Knowles.
\newblock {Infer.NET 2.4, Microsoft Research Cambridge}, 2010.

\bibitem[Moggi(1991)]{moggi-monads}
E.~Moggi.
\newblock Notions of computation and monads.
\newblock \emph{Inf.~Comput.}, 93\penalty0 (1):\penalty0 55--92, 1991.

\bibitem[Oles(1984)]{oles}
F.~J. Oles.
\newblock Type algebras, functor categories and block structure.
\newblock In \emph{Algebraic methods in semantics}. CUP, 1984.

\bibitem[Paige and Wood(2014)]{paige-icml-2014}
B.~Paige and F.~Wood.
\newblock A compilation target for probabilistic programming languages.
\newblock In \emph{ICML}, 2014.

\bibitem[Park et~al.(2008)Park, Pfenning, and
  Thrun]{parkpfenningthrun:sampling}
S.~Park, F.~Pfenning, and S.~Thrun.
\newblock A probabilistic language based on sampling functions.
\newblock \emph{ACM TOPLAS}, 31\penalty0 (1):\penalty0 171--182, 2008.

\bibitem[Patil et~al.(2010)Patil, Huard, and
  Fonnesbeck]{patilhuardfonnesbeck:pymc}
A.~Patil, D.~Huard, and C.~J. Fonnesbeck.
\newblock {PyMC}: {B}ayesian stochastic modelling in {P}ython.
\newblock \emph{Journal of Statistical Software}, 35, 2010.

\bibitem[Pitts(2000)]{pitts-catlogic}
A.~M. Pitts.
\newblock Categorical logic.
\newblock In \emph{Handbook of Logic in Computer Science}, volume~5. OUP, 2000.

\bibitem[PlotkinPower(2001)]{pp-adequacy}
G.~D.~Plotkin and J.~Power.
\newblock Adequacy for algebraic effects.
\newblock In \emph{FOSSACS}, 2001.

\bibitem[Power(2006)]{power-genericmodels}
J.~Power.
\newblock Generic models for computational effects.
\newblock \emph{TCS}, 364\penalty0 (2):\penalty0 254--269, 2006.

\bibitem[Ramsey and Pfeffer(2002)]{ramseypfeffer:stochasticlambda}
N.~Ramsey and A.~Pfeffer.
\newblock Stochastic lambda calculus and monads of probability distributions.
\newblock In \emph{POPL}, 2002.

\bibitem[Roy et~al.(2008)Roy, Mansinghka, Goodman, and
  Tenenbaum]{Roy-pp-nonparametric2008}
D.~M. Roy, V.~Mansinghka, N.~Goodman, and J.~Tenenbaum.
\newblock A stochastic programming perspective on nonparametric {B}ayes.
\newblock In \emph{ICML Workshop on Nonparametric Bayesian}, 2008.

\bibitem{scibor:haskell}
Adam Scibior, Zoubin Ghahramani, Andrew D. Gordon:
\newblock Practical probabilistic programming with monads. 
\newblock Haskell 2015.

\bibitem[Shan and Ramsey(2016)]{Shan-disintegration2016}
C.-C. Shan and N.~Ramsey.
\newblock Symbolic {B}ayesian inference by lazy partial evaluation, 2016.

\bibitem[Srivastava(1998)]{srivastava:borel}
S.~M. Srivastava.
\newblock \emph{A course on {B}orel sets}.
\newblock Springer, 1998.

\bibitem[{Stan Development Team}(2014)]{stan_software_2014}
{Stan Development Team}.
\newblock Stan: A {C}++ library for probability and sampling, version 2.5.0,
  2014.
\newblock URL \url{http://mc-stan.org/}.

\bibitem[Wood et~al.(2014)Wood, van~de Meent, and
  Mansinghka]{wood-aistats-2014}
F.~Wood, J.~W. van~de Meent, and V.~Mansinghka.
\newblock A new approach to probabilistic programming inference.
\newblock In \emph{AISTATS}, 2014.

\bibitem[Yang(2015)]{yang_DALI2015}
H.~Yang.
\newblock Program transformation for probabilistic programs, 2015.
\newblock Presentation at DALI.

\end{thebibliography}








\end{document}